\documentclass{article}

\usepackage{setspace,url}
\usepackage[left=1in,top=1in,right=1in,bottom=1in,dvips,letterpaper]{geometry}
\usepackage[algo2e,linesnumbered,vlined,ruled]{algorithm2e}
\usepackage{graphics,graphicx,epsf,subfigure,epstopdf}
\usepackage{comment}
\usepackage{times}
\usepackage{amsmath,amsxtra,amsfonts,amscd,amssymb,bm}
\usepackage{amsthm}
\usepackage[usenames]{color}
\usepackage[normalem]{ulem}
\usepackage{exscale}

\newtheorem{theorem}{Theorem}[section]
\newtheorem{lemma}[theorem]{Lemma}

\newtheorem{corollary}[theorem]{Corollary}
\newtheorem{definition}[theorem]{Definition}
\newtheorem{remark}[theorem]{Remark}

\newtheorem{assumption}[theorem]{Assumption}

\newcommand{\st}{\mathrm{ s.t. }}
\newcommand{\R}{\mathbb{R}}
\newcommand{\C}{\mathbb{C}}
\newcommand{\LL}{\mathbb{L}}

\newcommand{\cA}{\mathcal{A}}
\newcommand{\cM}{\mathcal{M}}
\newcommand{\cV}{\mathcal{V}}
\newcommand{\cP}{\mathcal{P}}

\newcommand{\cL}{\mathcal{L}}

\newcommand{\tr}{\mathrm{tr}}

\newcommand{\Rnn}{\R^{n \times n}}

\newcommand{\Rnk}{\R^{n \times p}}
\newcommand{\Rkk}{\R^{p \times p}}

\newcommand{\diag}{{\rm diag}}

\newcommand{\sspan}{\mathrm{span}}
\newcommand{\Diag}{\mathrm{Diag}}

\newcommand{\zz}{^{\mathrm{T}}}

\newcommand{\inv}{^{-1}}
\newcommand{\ff}{_{\mathrm{F}}}

\newcommand{\mR}{\mathbb{R}}
\newcommand{\lb}{\lambda}
\newcommand{\hlb}{\hat{\lambda}}
\newcommand{\Lb}{\Lambda}

\newcommand{\half}{\frac{1}{2}}

\newcommand{\Ld}{L^\dagger}
\newcommand{\yisan}{\frac{1}{3}}
\newcommand{\ersan}{\frac{2}{3}}
\newcommand{\roX}{\rho(X^*)}
\newcommand{\roY}{\rho(Y)}
\newcommand{\ltb}{\left(}
\newcommand{\rtb}{\right)}
\newcommand{\mI}{\mathcal{I}}

\newcommand{\mIs}{\mathcal{I}^*}

\newcommand{\Pj}{\mathbf{P}^{\perp}}

\newcommand{\lmax}{\lambda_{\max}}
\newcommand{\lmin}{\lambda_{\min}}

\newcommand{\Ph}{P^{\half}}
\newcommand{\FFVS}{F_{\phi}(V^*)}

\newcommand{\iprod}[2]{\left \langle #1, #2 \right \rangle }

\newcommand{\be}{\begin{equation}}
\newcommand{\ee}{\end{equation}}
\newcommand{\bee}{\begin{equation*}}
\newcommand{\eee}{\end{equation*}}
\newcommand{\bea}{\begin{eqnarray}}
\newcommand{\eea}{\end{eqnarray}}
\newcommand{\beaa}{\begin{eqnarray*}}
\newcommand{\eeaa}{\end{eqnarray*}}

\usepackage{booktabs}
\usepackage{multirow}
\newcommand{\minitab}[2][l]{\begin{tabular}{#1}#2\end{tabular}}

\allowdisplaybreaks

\begin{document}
\title{On the Analysis of the Discretized Kohn-Sham Density Functional Theory}


\author{Xin Liu\thanks{State Key Laboratory of Scientific and Engineering Computing, Academy of Mathematics and Systems Science, Chinese Academy of
Sciences, CHINA (liuxin@lsec.cc.ac.cn). Research supported in part
by NSFC grants 11101409, 11331012 and 91330115, and the National Center for
Mathematics and Interdisciplinary Sciences, CAS.}
\and Zaiwen Wen\thanks{Beijing International Center for Mathematical
Research, Peking University, CHINA (wenzw@math.pku.edu.cn).
Research supported in part by NSFC grants 11101274, 11322109 and 91330202.}
\and Xiao Wang\thanks{School of Mathematical Sciences, University of Chinese
Academy of Sciences, CHINA (wangxiao@ucas.ac.cn). Research supported in part by Postdoc Grant 119103S175, UCAS president grant Y35101AY00, and NSFC grant 11301505.}
\and Michael Ulbrich\thanks{Chair of Mathematical Optimization, Department of Mathematics, Technische Universit\"at M\"unchen, Boltzmannstr. 3, 85747 Garching b. M\"unchen, Germany.
(mulbrich@ma.tum.de).}
\and  Yaxiang Yuan\thanks{State Key Laboratory of Scientific and Engineering Computing,
Academy of Mathematics and Systems Science, Chinese Academy of
Sciences, CHINA (yyx@lsec.cc.ac.cn). Research supported in part
by NSFC grant 11331012.}
}

\maketitle

{\footnotesize
\textbf{Abstract.}
In this paper, we study a few theoretical issues in the discretized Kohn-Sham
(KS) density functional theory (DFT). The equivalence 
between either a local or global minimizer of the KS total energy
 minimization problem and the solution to the KS equation is established under
 certain assumptions. The nonzero charge densities of a strong local minimizer
 are shown to be bounded below by a positive constant uniformly.
We analyze the self-consistent field (SCF) iteration by formulating the KS equation as a fixed point map
with respect to the potential. The Jacobian of these fixed point maps is derived
explicitly. Both global and local convergence of the simple mixing scheme can be
established if the gap between the occupied states and unoccupied
states is sufficiently large.  This assumption can be relaxed if the charge density is computed using
 the Fermi-Dirac distribution and it is not required if there is no exchange
correlation functional in the total energy functional. Although our assumption on the gap  is very stringent and is almost
 never satisfied in reality, our analysis is still valuable for a better
 understanding of the KS minimization problem, the KS equation and the SCF iteration.


}

\section{Introduction}\label{sec:intro}
The Kohn-Sham  density functional theory in electronic structure calculations can be
formulated as either a total energy minimization
problem or a nonlinear eigenvalue problem. Using a suitable discretization
scheme whose   spatial degree of freedom is $n$, the electron wave
functions of $p$ occupied states can be approximated by a matrix  $X=[x_1,
\ldots, x_p] \in \Rnk$. The charge density of electrons associated with the occupied states
is defined as 
 \begin{eqnarray}\label{eq:rho}
\rho(X) := \diag(XX\zz),
\end{eqnarray}
where $\diag(A)$ denotes the vector containing the diagonal elements of the
matrix $A$. Let $\tr(A)$ be the trace of $A \in
\R^{n\times n}$, i.e., the sum of the diagonal elements of $A$. A commonly used discretized KS total energy
function has the form of 
\be \label{eq:KS-energy} E(X):=\frac{1}{4}\tr(X\zz L X) + \half\tr(X\zz V_{ion}  X)
 + \frac{1}{4} \rho^\top
L^\dagger \rho + \half e\zz \epsilon_{xc}(\rho), \ee
where   $L$ is a finite dimensional representation of the Laplacian operator, 
$V_{ion}$ is the ionic
pseudopotentials sampled on a suitably chosen Cartesian grid,  $L^\dagger$
 corresponds to the pseudo-inverse of $L$, $e$ is the column vector
of all ones and $\epsilon_{xc}(\rho)$ denotes the exchange correlation
energy functional.  The four terms in $E(X)$ describe the kinetic energy, local ionic potential energy,   Hartree
potential energy and exchange correlation
energy, respectively.

 The KS total energy minimization problem  solves
\be\label{prob:minKS}  \begin{aligned} \min_{X \in \Rnk} & \quad   E(X)  \\
  \st \; \; &\quad X\zz X = I.
  \end{aligned} \ee
  The orthogonality
constraints are imposed since  the wave
functions $X$ must be orthogonal to each other due to  physical constraints.
 It can be verified that the gradient of $E(X)$ with respect to $X$ is
 $\nabla E(X) = H(X)X$, where the Hamiltonian  $H(X)\in\Rnn$ is a matrix function  
 \be \label{eq:H} 
 H(X) := \half L + V_{ion}  +
\Diag(  L^\dagger \rho) + \Diag( \mu_{xc}(\rho)\zz e),\ee
 where $ \mu_{xc}(\rho) =\frac{ \partial \epsilon_{xc}}{\partial \rho} \in \R^{n \times
 n}$ and $\Diag(x)$  denotes a diagonal matrix
with $x$ on its diagonal.  The so-called KS equation is
\be \label{eq:KS} \begin{aligned}  H(X) X &=  X \Lb, \\
  X\zz X &=  I,
  \end{aligned} \ee
where $\Lb$  is a diagonal matrix consisting of $p$ smallest eigenvalues of
$H(X)$. The KS equation \eqref{eq:KS} is closely related to 
    the first-order optimality conditions for \eqref{prob:minKS} which are  
   the same as \eqref{eq:KS} except that the diagonal  matrix $\Lb$ consists of
   any $p$ eigenvalues of $H(X)$ rather than the $p$ smallest ones.

In this paper, we first study the relationship between the KS total energy
 minimization problem \eqref{prob:minKS} and the KS equation \eqref{eq:KS} under
 certain conditions. A simple counter example is provided to demonstrate that
 the solutions of these two problems are not necessarily the same.  The second-order optimality conditions of  \eqref{prob:minKS} are examined based
 on the assumption of the existence of the second-order derivative of the exchange correlation functional \cite{GaoYangMeza2009,Wen2013}.
 For a specialized exchange correlation functional, we prove that a global
 solution of \eqref{prob:minKS}  is a solution of 
 \eqref{eq:KS} if the gap between the $p$th and $(p+1)$st eigenvalues of the
 Hamiltonian $H(X)$ is sufficiently large.  The equivalence between a local
 minimizer of \eqref{prob:minKS} and the solution 
 \eqref{eq:KS} needs an additional assumption that the corresponding charge densities are all positive.
 For a strong local minimizer $X^*$ which is defined based on the second-order
 sufficient optimality conditions of \eqref{prob:minKS}, we show that the nonzero charge densities at  $X^*$ are bounded below
 by a positive constant uniformly.

Our second purpose is the analysis of the most widely used approach, the
self-consistent field (SCF) iteration, for solving the KS equation \eqref{eq:KS}.  The SCF iteration is   
 based on computing a sequence of
linear eigenvalue problems iteratively.   It is well known that the basic version of SCF iteration often converges slowly or fails to converge \cite{KoutechyBonacic1971} even with the help of various heuristics.  A convergence analysis of the SCF
iteration for solving the Hartree-Fock equations according to the optimal
damping algorithm (ODA) is established in \cite{CancesLeBris2000} and  an analysis of gradient-based algorithms  for the Hartree-Fock equations
is proposed in \cite{Levitt2012} using Lojasiewiscz inequality. 
 The interested reader is referred to 
 \cite{LeBris2005,Cances2000,Cances2001,Cancesetal2000,Cancesetal2003,Cancesetal2008,ChenZhouArXiv,ChenZhou2013,DaiZhou2011,DaiZhou2008,Kudinetal2002,Schneideretall2009}
 for discussion on ODA, the gradient-based algorithms and numerical analysis of DFT. 
  A condition  is identified in \cite{YangGaoMeza2009}
 such that
the SCF iteration is 
a contractive fixed point iteration under a specific form of the Hamiltonian
without involving any exchange correlation term.
  Global and local
convergence of the SCF iteration for general Kohn-Sham DFT is established in \cite{LiuWangWenYuan}
 from an optimization point of view.
   Their assumptions include that the second-order derivative of
  the exchange correlation energy functional is uniformly bounded from above and the gap between the $p$th and $(p+1)$st eigenvalues of the
 Hamiltonian $H(X)$ is sufficiently large. 
%

We improve the convergence results of the SCF iteration from the following
 three perspectives.  (i) The KS equation \eqref{eq:KS} 
is formulated as a nonlinear system of equations (fixed point maps) respect to either  the charge
density or potential. Applying the differentiability of spectral operators, the
Jacobian of these fixed point map is derived explicitly and analyzed.  
 (ii) Global convergence (i.e., convergence to a stationary point from any
 initial solution) of the simple mixing scheme can be established when there
 exists a gap between $p$th and $(p+1)$st eigenvalues of the
 Hamiltonian $H(X)$.  This assumption can be relaxed for local convergence
 analysis, i.e., convergence behavior if the initial point is selected in a neighborhood sufficiently
 close to the solution of \eqref{eq:KS}. If the charge density is computed using
 the Fermi-Dirac distribution, the assumption on the gap is not needed as long
 as a suitable step size for simple mixing is chosen. Our results
 requires much weaker conditions than the previous analysis in
 \cite{LiuWangWenYuan}. (iii) We propose two approximate Newton
 methods according to
the structure of the Jacobian of the fixed point maps. The second type of our
approaches is exactly the method of elliptic
preconditioner proposed in \cite{LinYang2012}. Preliminary convergence results
are also established for them.  
 Although our assumption on the gap between eigenvalues of the
 Hamiltonian  in the above three perspectives is very stringent and is almost
 never satisfied in reality, our analysis is still valuable for a better
 understanding of the KS equation and the SCF iteration.
 

The rest of this paper is organized as follows. A counter example between the
equivalence of the KS minimization and KS equation is presented in subsection
\ref{sec:count-example}. The optimality conditions of the KS minimization
problem under smoothness assumptions
on the exchange functional is provided  in subsection \ref{sec:opt}. The
necessary conditions for the equivalence between a local minimizer of
the KS minimization and the KS equation is established in subsection
\ref{sec:cond-local}. The corresponding analysis for a global minimizer is
established in subsection \ref{sec:cond-global}.  Lower bounds for the charge
density at  local minimizers are presented in subsection \ref{sec:low-bound}. In
subsection \ref{sec:prob},
we view the KS equation as fixed point maps with respect to the charge
density or potential. The Jacobian of these fixed
 point maps is presented in subsection \ref{sec:Jac}. 
In section \ref{sec:conv}, we establish both local and global convergence for
the SCF iteration with simple mixing schemes.  Two approximate Newton 
approaches and their convergence properties are  discussed  in section \ref{sec:App-Newton}.


\section{Equivalence Between the KS Total  Energy Minimization and the KS Equation}

\subsection{A Counter Example} \label{sec:count-example}

The following three-dimensional toy example shows that a solution of the KS equation
is not necessary a global optimal solution of the KS total  energy minimization
problem.
Let $n=3$, $p=1$ and choose \[  L=\begin{pmatrix}
1.4299 & -0.2839 &  -0.4056 \\ 
-0.2839 & 1.1874 &  0.2678 \\ 
-0.4056 & 0.2678 &  1.3826   \end{pmatrix}, \quad V_{ion}=0,  \mbox{
  and } \epsilon_{xc}(\rho) =0. \] 
  It can be verified numerically that  $X^*=\begin{pmatrix}0.3683 & -0.6188 &
    0.6939 \end{pmatrix}\zz$ is a global minimizer of \eqref{prob:minKS}.
     On the other hand,  we have 
  \[H(X^*) = \frac{1}{2}L + \Diag(  L^\dagger \rho(X^*)) = \begin{pmatrix}
0.9735 & -0.1419 &  -0.2028 \\ 
-0.1419 & 0.8955 &  0.1339 \\ 
-0.2028 & 0.1339 &  1.0569  \end{pmatrix},\]
 and  $X^*$ is an
  eigenvector associated with the second smallest eigenvalue of $H(X^*)$. Therefore,  the equivalence between  the KS
  total  energy minimization and the KS equation only holds under certain assumptions.
For this counter example, our assumptions in subsections \ref{sec:cond-local}
and \ref{sec:cond-global} do not
hold because the gap between the eigenvalues of $H(X^*)$ is $\delta=0.046$
 and it is smaller than $||\Ld||_2 = 1$. We should point out that the above
 example may not exist in the practice of DFT.


\subsection{Optimality Conditions Under Smoothness Assumptions on
$\epsilon_{xc}(\rho)$} \label{sec:opt}


The Lagrangian function of the minimization problem \eqref{prob:minKS} is
\[ \cL(X,\Lb) := E(X) - \frac{1}{2}\tr(\Lb (X\zz X-I)).
\]
Suppose $X$ is a local minimizer of \eqref{prob:minKS}. It follows from $X\zz X
= I$ that the linear independence constraint qualification is satisfied. Hence, there
exists a Lagrange multiplier $\Lb$ such that the first-order optimality
conditions hold:
\be \label{eq:KKT}  \nabla_X \cL(X,\Lb) = H(X)X - X\Lb =0 \mbox{ and } X\zz X =
I.\ee
Multiplying both sides of the first equality in \eqref{eq:KKT} by $X\zz$ and using  $X\zz X =I$, we
have $\Lambda = X\zz H(X) X$, which is a symmetric matrix. Note that $E(XQ) = E(X)$ and $H(XQ) = H(X)$ hold 
for any orthogonal matrix $Q\in\Rkk$.
Hence, if $X$ is a stationary point, any matrix in the set  
$\{XQ\mid Q\in\Rkk \mbox{ and } Q\zz Q=I \}$ is also a stationary point, 
and their  objective values are the same. Let
$\tilde{Q}\tilde{\Lambda}\tilde{Q}\zz$ be
the eigenvalue decomposition of  $X\zz H(X)X$ and $\tilde{X}:=X\tilde{Q}$. Then
the Lagrangian multiplier $\tilde{\Lambda}=\tilde{X}\zz H(\tilde{X}) \tilde{X}$
is a diagonal matrix whose entries are  the eigenvalues of $H(X)$.


 Let $\cL(\Rnk,\Rnk)$ denote the space of linear operators which map $\Rnk$ to
$\Rnk$. The Fr\'echet derivative of $\nabla E(X)$ is defined as the (unique)
function $\nabla^2 E: \Rnk \to \cL(\Rnk, \Rnk)$ such that
\[
\lim_{\|S\|\ff \to 0 } \frac{\| \nabla E(X+S) -\nabla E(X) - \nabla^2 E(X)
(S)\|\ff }{\|S\|\ff } =0.
\]
The next lemma shows an explicit form of the Hessian operator,
if the exchange correlation energy is second-order differentiable.
 \begin{lemma}[Lemma 2.1 in \cite{Wen2013}] \label{lemma:Hessian}Suppose that $\epsilon_{xc}(\rho(X))$ is twice
   differentiable with respect to $\rho(X)$.
Given a direction $S \in \Rnk$, the Hessian-vector product of $E(X)$ is
\be \label{eq:secondorderdef} \nabla^2E(X)[S] =   H(X) S  +2 \Diag \left(J(\rho) \diag(S  X \zz)  \right) X
, \ee
where 
\be \label{eq:def-J-rho} J(\rho) :=  L^\dagger +\partial \mu_{xc}(\rho) e. \ee
\end{lemma}

Consequently,  the second-order necessary and sufficient  optimality conditions can be obtained
from Theorems 12.5 and 12.6 in \cite{Nocedal2006}, respectively.

\begin{theorem}\label{thm:secondcon-tan}
1) Suppose that $X$ is a local minimizer of problem \eqref{prob:minKS} and  $\epsilon_{xc}(\rho(X))$ is twice
   differentiable with respect to $\rho(X)$. Then, for all $S\in\mathcal{T}(X)$, it holds
\begin{eqnarray}\label{eq:SO3}
 \tr(S\zz H(X)S- \Lb S\zz S) + 2\diag(XS\zz)\zz J\diag(XS\zz)\geq 0,
\end{eqnarray}
where $\Lb=X\zz H(X)X$ and 
\begin{eqnarray}\label{eq:tandef}
\mathcal{T}(X):=\{S\mid X\zz S + S\zz X =0\}.
\end{eqnarray}
 2) Suppose that $X\in\Rnk$
 satisfies \eqref{eq:KKT} with a symmetric matrix $\Lambda$
 and \eqref{eq:SO3} holds with a strict inequality for all $0\neq S\in\mathcal{T}(X)$. Then $X$ is a strict local
 minimizer for problem \eqref{prob:minKS}.
\end{theorem}
\begin{proof}
It follows from Theorem 12.5 in \cite{Nocedal2006} that the second-order necessary condition
for $X$ to be a local minimizer of \eqref{prob:minKS} is
\begin{eqnarray}\label{eq:SO}
\left\langle S, \nabla^2_{XX} \cL(X,\Lambda)[S] \right\rangle \geq 0,
\quad \mbox{ for all } S\in\mathcal{T}(X).
\end{eqnarray}
Using Lemma \ref{lemma:Hessian} and the fact that
  \[ \tr(X\zz \Diag(y) Z) = y\zz \diag(ZX\zz),\quad \mbox{ for all }
  X,Z\in\Rnk,\,y\in\mR^n, \]
we obtain 
\begin{eqnarray*}\label{eq:lag-sd}
 & & \left\langle S, \nabla^2_{XX} \cL(X,\Lambda)[S] \right\rangle 
 =  \tr(S\zz \nabla^2 E(X)[S] - \Lb S\zz S)\nonumber\\
&=& \tr\left(S\zz H(X)S + 2S\zz\Diag\left(J\diag(SX\zz)\right)X
- \Lb S\zz S\right)\nonumber\\
&=& \tr(S\zz H(X)S- \Lb S \zz S) + 2\diag(XS\zz)\zz
J \diag(XS\zz),
\end{eqnarray*}
which together with \eqref{eq:SO} yields \eqref{eq:SO3}. 
The second part is a direct application of Theorem 12.6 in \cite{Nocedal2006}.
\end{proof}

An equivalent formulation of the tangent space \eqref{eq:tandef} is 
\begin{eqnarray}\label{eq:tandef2}
\mathcal{T}(X) = \{S := XK +\Pj_{X}Z\mid K=-K\zz\in\Rkk,\, Z\in\Rnk\},
\end{eqnarray}
where $\Pj_{X}:=I-XX\zz$.
Hence, the second-order optimality conditions in Theorem \ref{thm:secondcon-tan}
can be presented in terms of an arbitrary $Z\in\Rnk$ similar to the analysis
of maximization of the sum of the trace ratio on the Stiefel Manifold 
in \cite{ZhangLi2013}.
\begin{theorem}\label{thm:secondcon-sp}
1) Suppose that $X$ is a local minimizer of problem \eqref{prob:minKS} and  $\epsilon_{xc}(\rho(X))$ is twice
   differentiable with respect to $\rho(X)$. Then for
all $Z\in\Rnk$, it holds
\begin{eqnarray}\label{eq:SO4}
&&\tr(Z\zz H(X) Z) + \tr(X\zz Z \Lb Z\zz X) - \tr(Z\zz X \Lb X\zz Z)
- \tr(Z\Lb Z\zz) \nonumber\\
&+& 2 \diag(XZ\zz\Pj_{X})\zz J \diag(XZ\zz\Pj_{X}) \geq 0.
\end{eqnarray}
 2) Suppose that $X\in\Rnk$
 satisfies \eqref{eq:KKT} with a symmetric matrix $\Lambda$
 and \eqref{eq:SO4} holds with a strict inequality for all $\Pj_{X}Z\neq 0$. Then $X$ is a strict local
 minimizer for problem \eqref{prob:minKS}.
\end{theorem}
\begin{proof}
Using \eqref{eq:KKT} and the definition of $\Pj_{X}$, we obtain  
$\Pj_{X}\Pj_{X}=\Pj_{X}$, $\Pj_{X}X=0$ and $\Pj_{X}H(X)X =0$. For  any $S=XK +\Pj_{X}Z$, it holds
\begin{eqnarray}\label{eq:term1}
\tr(S\zz H(X)S) &=& \tr(K\zz X\zz H(X)X K) + \tr(Z\zz\Pj_{X}H(X)\Pj_{X}Z)\nonumber\\
&=& \tr(K\zz \Lb K) + \tr(Z\zz H(X)Z) - \tr(Z\zz H(X) XX\zz Z) \nonumber\\
&=& \tr(K\zz \Lb K) + \tr(Z\zz H(X)Z) - \tr(Z\zz X \Lb X\zz Z).
\end{eqnarray}
It can be verified that $S\zz S = K\zz K +Z\zz\Pj_{X}Z$, which yields 
\begin{eqnarray}\label{eq:term2}
\tr(\Lb S\zz S) &=& \tr(K\zz K \Lb) + \tr(Z\zz Z \Lb) - \tr(Z\zz XX\zz Z \Lb)\nonumber\\
&=& \tr(K\zz \Lb  K) + \tr(Z\Lb Z\zz) -\tr(X\zz Z \Lb  Z\zz X),
\end{eqnarray}
where the last equality holds because of $K = - K\zz$.
Since it holds
\[\diag(XK\zz X\zz) = \frac{1}{2}(\diag(XK\zz X\zz) + \diag(XK X\zz))
= \frac{1}{2} \diag(X(K+K\zz)X\zz) = 0, \]
we obtain
\[\diag(XS\zz) = \diag(XK\zz X\zz) + 
\diag(XZ\zz\Pj_{X}) = \diag(XZ\zz\Pj_{X}),\]
which together with \eqref{eq:term1} and \eqref{eq:term2}  gives \eqref{eq:SO4}.
The proof of the second part follows directly from Theorem \ref{thm:secondcon-tan}.
\end{proof}

\subsection{Necessary Condition for Local Minimizers} \label{sec:cond-local}

In this subsection, we establish a necessary condition under which a local
minimizer of \eqref{prob:minKS} is  a solution of a modification of the KS equation
\eqref{eq:KS}.   Our discussion is restricted to a special exchange correlation functional 
\begin{eqnarray}\label{eq:exdef}
e\zz \epsilon_{xc}(\rho) = -\frac{3}{4}\gamma \rho\zz\rho^{\frac{1}{3}},
\end{eqnarray}
where $\gamma = 2 \left(\frac{3}{\pi} \right)^{1/3}$ and $\rho^{\frac{1}{3}}$
denotes the component-wise cubic root of the vector $\rho$.
 The next result shows that the charge density $\rho$  is bounded.
\begin{lemma}\label{lm:1}
Let $X\in\Rnk$ satisfy $X\zz X=I$, and $\rho$ be defined by 
\eqref{eq:rho}. We have
\begin{eqnarray}\label{eq:basicbound}
 0 \le  \rho_i \le 1, \mbox{   for all } i=1,\ldots,n.
\end{eqnarray}
\end{lemma}
\begin{proof} The inequality \eqref{eq:basicbound} holds from $X\zz X =I$ and the fact that
$ \rho_i = \sum_{j=1}^p X_{ij}^2$ for all $ i=1,\ldots,n$.
\end{proof}


Our analysis relies on the gap between the $p$th and $(p+1)$st eigenvalues of
$H(X)$. 
\begin{assumption}
\label{asmp:UWP}
Let $\lambda_1 \le \ldots \le \lambda_p \le \lambda_{p+1} \le \ldots \le
\lambda_n$ be the eigenvalues of a given symmetric matrix $H \in \Rnn$.
There exists a positive constant $\delta$ such that  
$\lambda_{p+1}-\lambda_p \ge \delta$.
\end{assumption}


Note that $E(X)$ may  not be second-order differentiable  since some components
$\rho_i(X)$ can be zero. Let $\mI$ be the
collection of indices of the nonzero components of $\rho(X)$, i.e.,
\be \label{eq:mI} \mI = \{i \mid \rho_i(X)\neq 0, i = 1,\ldots,n \}. \ee
Then the complement set $\bar{\mI}$ of $\mI$ is the set of indices of the zero
components of $\rho(X)$. Let $r$ be the cardinality of $\mI$. 
We have $r \ge p$ by the orthogonality of $X$.
If $\mI=\{\alpha_1, \ldots, \alpha_r\}$, 
we define the submatrices $X_\mI$
and $L_{\mI \mI}$  as 
\[ X_\mI = \begin{pmatrix} X_{\alpha_1, 1}, \ldots,  X_{\alpha_1, p} \\
\ldots \\
X_{\alpha_r, 1}, \ldots,  X_{\alpha_r, p}
\end{pmatrix}, \mbox{ and } L_{\mI \mI} =\begin{pmatrix} L_{\alpha_1, 1}, \ldots,
  L_{\alpha_1, \alpha_r} \\
\ldots \\
L_{\alpha_r, 1}, \ldots,  L_{\alpha_r, \alpha_r}
\end{pmatrix}. \]
The notations $(V_{ion})_{\mI \mI}$,  $L^\dagger_{\mI \mI}$,  $H_{\mI
\mI}(X)$ and  $\Lb_{\mI \mI}$ are
defined similar to  $L_{\mI \mI}$. 

The following theorem shows that a local minimizer $X^*$ of the KS total energy minimization 
\eqref{prob:minKS}  is a solution of KS equation \eqref{eq:KS}
if all rows of $X^*$ are nonzero and  Assumption \ref{asmp:UWP}
holds with a sufficiently large gap $\delta$.
\begin{theorem}\label{thm:localcon}
Suppose that  $X^*$ is a local
minimizer of \eqref{prob:minKS} using \eqref{eq:exdef} and $\Lambda^* = (X^*)^\top
H(X^*)X^*$ is a diagonal matrix.
 Let  $\mathcal{I^*}$ be the
index set of $X^*$ defined as \eqref{eq:mI}. If Assumption \ref{asmp:UWP} holds
at $H(X^*)$ with a constant $\delta$ satisfying 
\begin{eqnarray}\label{eq:cond3}
\delta > 2\left(||\Ld||_2 - \frac{\gamma}{3}\right),
\end{eqnarray}
  then it holds
\be\label{eq:KSs1}  \begin{aligned}  
H_{\mIs\mIs}(X^*) X^*_{\mIs} &=  X^*_{\mIs}  \Lb^*, \\
 (X^*_{\mIs}) \zz X^*_{\mIs} &=  I,
\end{aligned} \ee
and the diagonal of $\Lb^*$ consists of
the $p$ smallest eigenvalues of $H_{\mIs\mIs}(X^*)$.
\end{theorem}
\begin{proof}
It can be verified that $X^*$ is a local minimizer of the restricted problem 
\be\label{prob:minKS-D}  \begin{aligned} \min_{X \in \Rnk} & \quad   E(X)  \\
  \st \; \; &\quad X\zz X = I, \; X_{\bar \mIs}=0.
  \end{aligned} \ee
Hence, $X^*_{\mIs}$ is a local minimizer of the reduced problem
\be\label{eq:obj-X-rd}
\begin{aligned}
\min\limits_{\hat{X}\in\mR^{r\times p}} 
& \quad \hat{E}(\hat{X}) := \frac{1}{4}\tr( \hat X\zz L_{\mIs \mIs} \hat X) +
\half\tr(\hat X\zz (V_{ion})_{\mIs \mIs} \hat X)
 + \frac{1}{4} \rho(\hat X)^\top
L^\dagger_{\mIs \mIs} \rho(\hat X) -\frac{3}{4}\gamma \rho(\hat X)\zz\rho(\hat X)^{\frac{1}{3}}, \\
 \st \quad & \quad \hat{X}\zz \hat{X} = I.
\end{aligned} \ee
 The structure of the energy functional $E(X)$ implies $\nabla \hat
 E(X^*_{\mIs}) = H_{\mIs\mIs}(X^*)X^*_{\mIs} $ and $ (X^*_{\mIs})\zz H_{\mIs\mIs}(X^*_{\mIs}) X^*_{\mIs} =
\Lb^*$. These facts together with 
 the first-order optimality of \eqref{eq:obj-X-rd} at $X^*_{\mIs}$ yield
 \eqref{eq:KSs1}. 
 
 It is obvious that the diagonal entries of $\Lb^*$ are the eigenvalues of $H_{\mIs\mIs}(X^*)$. Suppose that they are not the $p$ smallest eigenvalues of $H_{\mIs\mIs}(X^*)$.  
 For convenience, we denote the eigenvalues of
$H_{\mIs\mIs}(X^*)$ in an ascending order as $\hlb_1\leq ...\leq \hlb_r$ and
their corresponding eigenvectors are $u_i$, $i=1,\ldots,r$, where
$r = |\mIs|$. Let $x_i$, $1 \le i\le p$, be the $i$th column for $X^*_{\mIs}$.  
Without loss of generality, let $x_1$ be  associated with 
an eigenvalue greater than $\hlb_p$, and $u_i$ ($i\leq p$) be an eigenvector associated with
an eigenvalue less than or equal to $\hlb_p$ but not be a column of $X^*_{\mIs}$.
The Assumption \ref{asmp:UWP} implies that  $u_i\notin \sspan{\{X^*_{\mIs}\}}$. Let $V$ be a matrix whose columns satisfy
\[v_j = 
\begin{cases}
u_i & \mbox{if }j=1,\\
x_j & \mbox{if } j=2,\ldots,p.\\
\end{cases} \] 
 Since the function  $\hat{E}(\hat{X})$ is
twice differentiable at $X^*_{\mIs}$ according to the definition of $\mIs$.
Therefore, an application of Theorem \ref{thm:secondcon-sp} gives
\bea\label{eq:localderive}
 \Delta&:=&\tr(V\zz H_{\mIs\mIs}(X^*_{\mIs}) V) + \tr((X^*_{\mIs})\zz V \Lb^* V\zz X^*_{\mIs})
 - \tr(V\zz X^*_{\mIs} \Lb^* (X^*_{\mIs})\zz V)
- \tr(V\Lb^* V\zz) \nonumber\\
&& + 2 \diag(X^*_{\mIs}V\zz \Pj_{X^*_{\mIs}})\zz\left(\Ld_{\mIs\mIs} - \frac{\gamma}{3}\Diag\left(\rho(X^*_{\mIs})^{-\frac{2}{3}}\right)\right)
\diag(X^*_{\mIs}V\zz\Pj_{X^*_{\mIs}}) \nonumber \\ &\ge& 0. \eea
It follows from that $V$ is  an orthonormal eigenbasis of $H_{\mIs\mIs}(X^*_{\mIs})$
and  Assumption \ref{asmp:UWP} that
\begin{eqnarray} \label{eq:gap-LV}
\tr(V\zz H_{\mIs\mIs}(X^*) V) - \tr((X^*_{\mIs})\zz H_{\mIs\mIs}(X^*)X^*_{\mIs})
\leq \hlb_i-\hlb_{p+1} \leq -\delta.
\end{eqnarray}
Since  $u_i\notin \sspan{\{X^*_{\mIs}\}}$, we obtain
\begin{eqnarray}
(X^*_{\mIs})\zz V =  V\zz X^*_{\mIs} &=&  I - e_1e_1\zz, \\
X^*_{\mIs}V\zz\Pj_{X^*_{\mIs}} &=&  x_1u_i\zz,
\end{eqnarray}
which further give
\begin{eqnarray}\label{eq:localderive2}
\Delta&=& \tr(V\zz H_{\mIs\mIs}(X^*_{\mIs}) V) - \tr(\Lb^*) +2 \diag(x_1u_i\zz)\zz\left(\Ld_{\mIs\mIs} - \frac{\gamma}{3}\Diag\left(\rho(X^*_{\mIs})^{-\frac{2}{3}}\right)\right)
\diag(x_1u_i\zz)\nonumber\\
&\leq& -\delta + 
 2 \max\left\{\lb_{\max}\left(\Ld_{\mIs\mIs} - \frac{\gamma}{3}\Diag\left(\rho(X^*_{\mIs})^{-\frac{2}{3}}\right)
\right),\,0\right\}
\nonumber\\
&\leq& -\delta + 2 \max\left\{\lb_{\max}\left(\Ld_{\mIs\mIs} - \frac{\gamma}{3}I\right),\,0\right\}
\nonumber \\
&\leq&  -\delta + 2 \max\left\{\ltb||\Ld_{\mIs\mIs} ||_2 -
\frac{\gamma}{3}\rtb,\,0\right\} \nonumber\\
&<& 0,
\end{eqnarray}
where the first inequality uses \eqref{eq:gap-LV} and the fact that $||\diag(x_1u_i\zz)||_2^2\leq 1$, the second inequality follows from $\rho\in[0,1]$,  the third inequality uses the fact that
$||\Ld_{\mIs\mIs} ||_2\leq ||\Ld||_2$ since the largest/smallest eigenvalue of 
a matrix is no less/greater than the largest/smallest eigenvalue of
its any principal submatrix, and the last inequality \eqref{eq:localderive2} is
due to \eqref{eq:cond3}. However,  
 \eqref{eq:localderive2} is a contradiction to \eqref{eq:localderive}. This completes the proof.
\end{proof}

\subsection{Necessary Condition for Global Minimizers}\label{sec:cond-global}

In this subsection, we consider whether a global
minimizer of \eqref{prob:minKS}  is a solution of the KS equation
\eqref{eq:KS} under the  exchange correlation functional 
\eqref{eq:exdef}. We first show the following inequality.
\begin{lemma} \label{lemma:ab-ineq}
It holds for  all $a,b\in [0,1]$  that
\[
(a-b)^2 (3 a^2+2 ab+b^2)= 3a^4-4 a^3b+b^4 \ge \frac{2}{3} (a^3-b^3)^2.
\]
\end{lemma}
\begin{proof}
The inequality holds for $a=0$ or $b=0$. Consider the case on $a\ge b>0$. Introducing the variable $t=b/a \in (0,1]$ yields
\[
  a^4 (3-4t+t^4)- \frac{2}{3} a^6 (1-t^3)^2 \ge a^6 f(t),
\]
where $f(t)=3-4t+t^4- \frac{2}{3}  (1-t^3)^2$.  Since 
$f'(t)= (t^3-1)(4-4 t^2)\le 0$ for all $t\in[0,1]$, we have
$f(t)\ge f(1)=0$ for all $t\in[0,1]$, and then the inequality is proved. The case 
 on $b\ge a>0$ can be proved in a similar fashion.
\end{proof}

The next theorem establishes the equivalence based on estimating the
difference of total energy function values.
\begin{theorem}\label{thm:globalcon}
Suppose that $X^*$ is a global  
minimizer of \eqref{prob:minKS} using \eqref{eq:exdef}.   If Assumption
\ref{asmp:UWP} holds at $H(X^*)$ with a constant $\delta$ satisfying 
\begin{eqnarray}\label{eq:cond1}
\delta > p \ltb||\Ld||_2 - \frac{\gamma}{3}\rtb,
\end{eqnarray}
then $X^*$ must be an orthonormal eigenbasis of $H(X^*)$
corresponding to its $p$ smallest eigenvalues, namely, a solution of 
the KS equation \eqref{eq:KS}.
\end{theorem}
\begin{proof}
  Suppose that $X^*$ is not but $Y$ is an orthonormal eigenbasis of $H(X^*)$
corresponding to its $p$ smallest eigenvalues.   Since $X^*$ must be an
orthonormal eigenbasis of $H(X^*)$ and using Assumption \ref{asmp:UWP}, we  have
\be \label{eq:red-HYX}
\Delta H(Y,X^*):=\tr(Y\zz H(X^*) Y) - \tr( (X^*)\zz H(X^*)X^*) \leq
\lb_p(H(X^*))-\lb_{p+1}(H(X^*))
\leq -\delta.
\ee
Applying Lemmas \ref{lm:1} and \ref{lemma:ab-ineq} gives 
\be\label{eq:rho-ineq1} \sum\limits_{i=1}^n
\ltb\rho(Y)_i^{\yisan} -\rho(X^*)_i^{\yisan}\rtb^2 \left(3\roY_i^{\ersan}
+2\roY_i^{\yisan}\roX_i^{\yisan} + \roX_i^{\ersan}\right) \ge \frac{2}{3} ||\roY
- \roX||_2^2.  \ee
It follows from Lemma \ref{lm:1} that
\bea\label{eq:rho-ineq2} \|\rho(Y) - \rho(X^*) \|^2 
&\le& 
 (1- \rho(Y))\zz \rho(X^*) +  (1- \rho(X^*))\zz \rho(Y) \\
 &\le& 1\zz \rho(X^*) + 1\zz \rho(Y) = \tr(X X\zz) + \tr(Y Y\zz) \nonumber \\
 &=& 2 p. \nonumber\eea
Using the relationship $\tr(Y\zz \Diag(\Ld \rho(X^*))Y) = \rho(Y)\zz \Ld \rho(X^*)$,
 the inequalities \eqref{eq:red-HYX}, \eqref{eq:rho-ineq1} and
 \eqref{eq:rho-ineq2}, and the assumption \eqref{eq:cond1}, 
we
obtain 
\begin{eqnarray}\label{eq:dE}
\Delta E(Y,X^*) &=& E(Y) - E(X^*) \nonumber \\
&=& \frac{1}{2}\Delta H(Y,X^*) 
+ \frac{1}{4}\left(\rho(Y)\zz \Ld \rho(Y) - \rho(X^*)\zz \Ld \rho(X^*) \right) 
- \frac{3\gamma}{8}\left(\rho(Y)\zz \rho(Y)^{\frac{1}{3}} - \rho(X^*)\zz \rho(X^*)^{\frac{1}{3}}\right)
\nonumber\\
&&  -\frac{1}{2} \tr(Y\zz
\Diag(\Ld \rho(X^*)-\gamma\rho(X^*)^{\frac{1}{3}})Y) +\frac{1}{2} \tr(X\zz
\Diag(\Ld \rho(X^*)-\gamma\rho(X^*)^{\frac{1}{3}})X^*)  \nonumber \\
 &=& \frac{1}{2}\Delta H(Y,X^*) 
+ \frac{1}{4}\left(\rho(Y)\zz \Ld \rho(Y) - \rho(X^*)\zz \Ld \rho(X^*) \right) 
- \frac{3\gamma}{8}\left(\rho(Y)\zz \rho(Y)^{\frac{1}{3}} - \rho(X^*)\zz \rho(X^*)^{\frac{1}{3}}\right)
\nonumber\\
&& - \frac{1}{2}\left(\rho(Y)\zz \Ld \rho(X^*) - \rho(X^*)\zz \Ld \rho(X^*) \right)
 + \frac{1}{2}\gamma \left(\rho(Y)\zz \rho(X^*)^{\frac{1}{3}} - \rho(X^*)\zz \rho(X^*)^{\frac{1}{3}}\right)\nonumber \\
&=& \frac{1}{2}\Delta H(Y,X^*) 
+ \frac{1}{4}(\rho(Y)-\rho(X^*))\zz \Ld 
(\rho(Y)-\rho(X^*)) \nonumber\\
&& - \frac{\gamma}{8}\sum\limits_{i=1}^n
\ltb\rho(Y)_i^{\yisan} -\rho(X^*)_i^{\yisan}\rtb^2 \left(3\roY_i^{\ersan}
+2\roY_i^{\yisan}\roX_i^{\yisan} + \roX_i^{\ersan}\right) \nonumber \\
&\le& - \frac{\delta}{2} + 
\ltb\frac{||\Ld||_2}{4}-\frac{\gamma}{12}\rtb ||\roY - \roX||_2^2 \nonumber \\
&\leq& - \frac{\delta}{2} + \ltb\frac{||\Ld||_2}{4}-\frac{\gamma}{12}\rtb (2p)
\nonumber \\
&<& 0, \nonumber
\end{eqnarray} 
which  is a contradiction to the fact that $X^*$ is a global minimizer. This
completes the proof.
\end{proof}


\begin{remark}
When the exchange correlation function $\epsilon_{xc}(\rho)$ is equal to zero, our condition \eqref{eq:cond1} becomes
$\delta > p||\Ld||_2$, which is much weaker than the condition
$\delta > 12p\sqrt{n}||\Ld||_2$
in Theorem 1 of \cite{LiuWangWenYuan}.
\end{remark}

\subsection{Lower Bounds for the Charge Density of Local
Minimizers} \label{sec:low-bound}

The exchange correlation energy functional is twice differentiable  if all
components of $\rho(X)$ are positive. However, the second-order derivative
may not be bounded at an arbitrary point $X$. In this subsection, we provides a few lower bounds for the charge density at certain types local
minimizers. These properties are useful for our analysis on the KS equation.

%


Traditionally, a point $x^*$ is called a strong local minimizer \cite{Shapiro2000,Lewis2013} of 
a function $f: \mR^n \mapsto \mR$, if there exists a constant $\kappa>0$ and a neighborhood
$U$ of $x^*$ such that the inequality
\begin{eqnarray} \label{eq:strong-f} 
f(x)\geq f(x^*) + \kappa ||x-x^*||_2^2
\end{eqnarray}
holds for any $x\in U$. Here, we define a strong local minimizer based on
the second-order optimality conditions. 
\begin{definition}\label{dy:secondcon-tan-strong}
A point $X^*$ is called a strong local minimizer of \eqref{prob:minKS} using  
\eqref{eq:exdef} if and only if $X^*_{\mIs}$ is local minimizer of \eqref{eq:obj-X-rd}
and there exists a constant $\kappa>0$ such that, for all $Z \in \Rnk$,
\begin{eqnarray}\label{eq:SO4-strong}
&&\tr(Z\zz H_{\mIs\mIs}(X^*_{\mIs}) Z) + \tr( (X^*_{\mIs})\zz Z \Lb^* Z\zz X^*_{\mIs}) 
- \tr(Z\zz X^*_{\mIs} \Lb^* (X ^*_{\mIs})\zz Z)
- \tr(Z\Lb^* Z\zz) \nonumber\\
&+& 2 \diag( (X^*_{\mIs})Z\zz\Pj_{X^*_{\mIs}})\zz\left(\Ld_{\mIs\mIs} -
\frac{\gamma}{3}\Diag\left(\rho(X^*_{\mIs})^{-\frac{2}{3}}\right)\right)
\diag( (X^*_{\mIs})Z\zz \Pj_{X^*_{\mIs}}) \geq \kappa \|Z\|_F^2,
\end{eqnarray}
where $\Lb^*= (X^*_{\mIs})\zz H_{\mIs\mIs}(X^*)X^*_{\mIs}$ and $\mathcal{I^*}$
is the
index set of $X^*$ defined as \eqref{eq:mI}. 
\end{definition}

%

Our condition \eqref{eq:SO4-strong} is weaker than \eqref{eq:strong-f} applying
to problem \eqref{prob:minKS} when the total energy $E(X)$ is twice
differentiable.  
The next result shows that the charge densities at a strong local minimizer are 
bounded below uniformly if they are positive. 
\begin{theorem}\label{thm:bound}
  Suppose that $L$ is positive semidefinite and $X^*$ is a strong local minimizer of \eqref{prob:minKS} satisfying Definition 
\ref{dy:secondcon-tan-strong}.
 Let  
 \begin{eqnarray}\label{eq:Ldef}
\bar{c}:=\min\{1,\,c_1,\,...,\,c_n\} \mbox{ and } c_i:= \min\limits_{j\neq i}\left(\frac{\gamma}{3(L^\dagger_{ii}-2L^\dagger_{ij}
+L^\dagger_{jj})}\right)^{\frac{3}{2}}.
\end{eqnarray}
 Then it holds:
\begin{eqnarray}\label{eq:bound}
  \mbox{ for any }  i\in\{1,2,...,n\}, \qquad \rho_i(X^*)\in[0,\bar{c}) \quad \Rightarrow \quad
\rho_i(X^*) = 0.
\end{eqnarray}
\end{theorem}
\begin{proof}
For convenience, we denote $\rho^*_{\mIs} = \rho(X^*_{\mIs})$.
If there exists a row $j$ in $X^*_{\mIs}$ such that either $1$ or $-1$ is an
entry of this row, then this row has only 
one nonzero entry  according to the orthonormality of $X^*_{\mIs}$. Hence,
$(\rho^*_{\mIs})_j = 1$ and \eqref{eq:bound} holds at $j$.

 We next consider the components in the set $ \mathcal{J}:=\{j\mid 
 j \in \mIs \mbox{ and } |(X^*_{\mIs})_{js}| < 1,\, s=1,\ldots,p\}$.
For any given $j\in \mathcal{J}$, there exists a 
nonzero entry, denoted as $(X^*_{\mIs})_{js}$, in the $j$-th row of $X^*_{\mIs}$.  Since $|(X^*_{\mIs})_{js}|<1$,
there exists at least another nonzero entry, denoted as $(X^*_{\mIs})_{is}$,  in the $s$-th column of $X^*_{\mIs}$
due to the orthonormality of $X^*_{\mIs}$.
For simplicity, let $x_l$, $l=1,...,p$, be the $l$-th column of $X^*_{\mIs}$ and
set $r= |\mIs|$,  $x_{js} = (X^*_{\mIs})_{js}$ and $x_{is} = (X^*_{\mIs})_{is}$.
Define a vector $z\in\mR^r$ whose $l$-th component ($l=1,\ldots,p$) is 
\begin{eqnarray}
z_l =\left\{
\begin{array}{cc}
\frac{x_{is}}{\sqrt{x_{is}^2+x_{js}^2}}, & \mbox{if}\, l=j;\\
\frac{-x_{js}}{\sqrt{x_{is}^2+x_{js}^2}}, & \mbox{if}\, l=i;\\
0, & \mbox{otherwise.}
\end{array}
\right.
\end{eqnarray}
A short calculation gives $||z||_2=1$, $z\zz x_s =0$ and 
\begin{eqnarray}\label{eq:zx}
\diag(z x_s\zz) = \frac{x_{is}x_{js}}{\sqrt{x_{is}^2+x_{js}^2}} e_{(j,-i)},
\end{eqnarray}
where $e_{(j,-i)}\in \mR^r$ has $1$ on its $j$-th entry, $-1$ on its $i$-th entry and $0$ elsewhere.

 For $a\in[0,1]$, let $Z_a\in\Rnk$ be a matrix whose $s$-th column is $a
z+\sqrt{1-a^2}x_s$ and all other columns are zero.
Without loss of generality, let  $\hlb_1\leq ...\leq \hlb_r$ 
be the eigenvalues of $H_{\mIs\mIs}(X^*)$ 
in the ascending order, and $x_s$ be an eigenvector of $H_{\mIs\mIs}(X^*)$
associated with $\hlb_s$, $s\in\{1,...,r\}$.
Then, we obtain 
\begin{eqnarray}\label{eq:Z1}
\tr(Z_a\zz H_{\mIs\mIs}(X^*) Z_a) \leq a^2\hlb_r + (1-a^2)\hlb_s ,\\
\label{eq:Z2} \tr(Z_a\Lb^* Z_a\zz) = \tr(\Lb^* Z_a\zz Z_a) = \hlb_s ,
\end{eqnarray}
which yields 
\begin{eqnarray}\label{eq:Z12}
\tr(Z_a\zz H_{\mIs\mIs}(X^*) Z_a) 
- \tr(Z_a\Lb^* Z_a\zz) \leq a^2\hlb_r + (1-a^2)\hlb_s - \hlb_s 
= a^2(\hlb_r - \hlb_s).
\end{eqnarray}
The definition of $Z_a$  gives
\begin{eqnarray}\label{eq:ZX}
(Z_a\zz X^*_{\mIs})_{pq}=\left\{
\begin{array}{cc}
a z\zz x_q, & \mbox{if}\, p=s,\,q\neq s;\\
\sqrt{1-a^2} , & \mbox{if}\, p=s,\,q = s;\\
0, & \mbox{otherwise.}
\end{array}
\right.
\end{eqnarray}
Hence, we have
\begin{eqnarray}\label{eq:Z3}
\tr((X^*_{\mIs})\zz Z_a \Lb^* Z_a\zz X^*_{\mIs}) 
&=& \tr (\Lb^* Z_a \zz X^*_{\mIs} (X^*_{\mIs})\zz Z_a)
= \hlb_s \left(\sum\limits_{q=1,q\neq s}^p a^2 (z\zz x_q)^2 + (1-a^2)\right)\nonumber\\
&=& \hlb_s \left(\sum\limits_{q=1}^p a^2 (z\zz x_q)^2 + (1-a^2) - a^2(z\zz x_s)^2\right)\nonumber\\
&=& \hlb_s (1 + a^2 ||z\zz X^*_{\mIs}||_2^2  -a^2)
= a^2 \hlb_s||z\zz X^*_{\mIs}||_2^2   + (1-a^2)\hlb_s.
\end{eqnarray}
and 
\begin{eqnarray}\label{eq:Z4}
\tr(Z_a\zz X^*_{\mIs} \Lb^* (X^*_{\mIs})\zz Z_a)
&=&  \left(\sum\limits_{q=1,q\neq s}^p a^2 (z\zz x_q)^2 \hlb_q  + (1-a^2)\hlb_s\right)\nonumber\\
&\geq&  \left(\sum\limits_{q=1,q\neq s}^p a^2 (z\zz x_q)^2 \hlb_1  
+ (1-a^2)\hlb_s\right)\nonumber\\
&=&  \left(\sum\limits_{q=1}^p a^2 (z\zz x_q)^2 \hlb_1  
+ (1-a^2)\hlb_s \right)\nonumber\\
&=& a^2 \hlb_1||z\zz X^*_{\mIs}||_2^2 +  (1-a^2) \hlb_s.
\end{eqnarray}
Combining \eqref{eq:Z3} and \eqref{eq:Z4} together yields
\begin{eqnarray}\label{eq:Z34}
&& \tr((X^*_{\mIs})\zz Z_a \Lb^* Z_a \zz X^*_{\mIs}) 
 - \tr(Z_a\zz X^*_{\mIs} \Lb^* (X^*_{\mIs})\zz Z_a) \nonumber\\
&\leq& 
(a^2 \hlb_s||z\zz X^*_{\mIs}||_2^2   + (1-a^2)\hlb_s) -  (a^2 \hlb_1||z\zz X^*_{\mIs}||_2^2 +  (1-a^2) \hlb_s)
= a^2(\hlb_s - \hlb_1)||z\zz X^*_{\mIs}||_2^2\nonumber\\
&\leq& a^2(\hlb_s - \hlb_1).
\end{eqnarray}
The equality \eqref{eq:zx} gives
\begin{eqnarray}\label{eq:Z5}
\diag(X^*_{\mIs}Z_a\zz\Pj_{X^*_{\mIs}}) = a \diag(x_s z\zz)
=\frac{ax_{is}x_{js}}{\sqrt{x_{is}^2+x_{js}^2}} e_{(j,-i)}.
\end{eqnarray}
Let $Z_a$ with $a = \sqrt{\frac{\kappa}{\hlb_r - \hlb_1}}$.
 Using \eqref{eq:Z12} and \eqref{eq:Z34}, we have
\begin{eqnarray} \label{eq:SO4-strong-use-ineq1}
&&\tr(Z_a\zz H_{\mIs\mIs}(X^*_{\mIs}) Z_a) 
+ \tr((X^*_{\mIs})\zz Z_a \Lb_{\mIs} Z_a\zz X^*_{\mIs}) \nonumber\\
&&
- \tr(Z_a\zz X^*_{\mIs} \Lb_{\mIs} (X^*_{\mIs})\zz Z_a)
- \tr(Z_a\Lb_{\mIs} Z_a\zz) \leq a^2(\hlb_r - \hlb_1) = \kappa.
\end{eqnarray}

It follows from our definition of strong local minimizers that
\begin{eqnarray}\label{eq:SO4-strong-use}
&&\tr(Z_a\zz H_{\mIs\mIs}(X^*) Z_a) 
+ \tr((X^*_{\mIs})\zz Z_a \Lb^* Z_a\zz X^*_{\mIs}) - \tr(Z_a\zz X^*_{\mIs} \Lb^* (X^*_{\mIs})\zz Z_a)
- \tr(Z_a\Lb^* Z_a\zz) \nonumber\\
&+&  2 \diag(X^*_{\mIs}Z_a\zz\Pj_{X})\zz\left(\Ld_{\mIs\mIs} - \frac{\gamma}{3}\Diag\left((\rho^*_{\mIs})^{-\frac{2}{3}}\right)\right)
\diag(X^*_{\mIs}Z_a\zz\Pj_{X^*_{\mIs}}) \geq \kappa,
\end{eqnarray}
which together with \eqref{eq:SO4-strong-use-ineq1}
 gives \begin{eqnarray}\label{eq:SO5}
\diag(X^*_{\mIs}Z_a\zz\Pj_{X^*_{\mIs}})\zz
\left(\Ld_{\mIs\mIs} - \frac{\gamma}{3}\Diag\left((\rho^*_{\mIs})^{-\frac{2}{3}}\right)\right)
\diag(X^*_{\mIs}Z_a\zz\Pj_{X^*_{\mIs}}) \geq 0.
\end{eqnarray}
Substituting \eqref{eq:Z5} into \eqref{eq:SO5}, we obtain 
\begin{eqnarray}\label{eq:SO51}
e_{(j,-i)}\zz
\left(\Ld_{\mIs\mIs} - \frac{\gamma}{3}\Diag\left((\rho^*_{\mIs})^{-\frac{2}{3}}\right)\right)
e_{(j,-i)} \geq 0.
\end{eqnarray}
Expending the terms of \eqref{eq:SO51} yields
\begin{eqnarray}\label{eq:ineq1}
(\Ld_{\mIs\mIs})_{jj}-2(\Ld_{\mIs\mIs})_{ji}+(\Ld_{\mIs\mIs})_{ii}
- \frac{\gamma}{3}(\rho^*_{\mIs})_j^{-\frac{2}{3}}
- \frac{\gamma}{3}(\rho^*_{\mIs})_i^{-\frac{2}{3}}\geq 0,
\end{eqnarray}
which implies
\begin{eqnarray}\label{eq:ineq2}
(\Ld_{\mIs\mIs})_{jj}-2(\Ld_{\mIs\mIs})_{ji}+(\Ld_{\mIs\mIs})_{ii}
\geq \frac{\gamma}{3}(\rho^*_{\mIs})_j^{-\frac{2}{3}}.
\end{eqnarray}
Therefore, we obtain 
\begin{eqnarray}\label{eq:ineq3}
(\rho^*_{\mIs})_j\geq \left(\frac{\gamma}{3((\Ld_{\mIs\mIs})_{jj}
-2(\Ld_{\mIs\mIs})_{ji}+(\Ld_{\mIs\mIs})_{ii})}\right)^{\frac{3}{2}} \geq c_j,
\end{eqnarray}
where $c_j$ is defined in \eqref{eq:Ldef}.
 Similarly, we can prove \eqref{eq:ineq3} holds for any $j\in \mathcal{J}$.
This completes the proof.
\end{proof}

\section{Analysis of the KS Equation}
\subsection{Formulating the KS Equation as a Fixed Point Map}\label{sec:prob}


The KS equation \eqref{eq:KS}  constitutes a nonlinear system with respect to $X$.
 Note that the Hamiltonian matrix \eqref{eq:H} is a symmetric matrix function
with respect to $\rho$ as 
\be \label{eq:H-rho} \hat H(\rho):= \half L + V_{ion}  +
\Diag(L^\dagger \rho) + \Diag(\mu_{xc}(\rho)\zz e),\ee 
and the KS equation becomes
\be \label{eq:KS1} \left\{ \begin{aligned} \hat H(\rho) X &=  X \Lb, \\
  X\zz X &=  I,
  \end{aligned} \right.\ee
where $X\in\Rnk$ and $\Lb\in\Rkk$ is a diagonal matrix consisting of the $p$ smallest eigenvalues of
$\hat H(\rho)$. 
 The eigenvalue  decomposition of $\hat H(\rho)$ is determined once
 $\rho$ is given. Hence, we can write $X$ as $X(\rho)$ to reflect the dependence
 on $\rho$ and the KS equation \eqref{eq:KS} can be viewed as a system of nonlinear
equations with respect to the charge density $\rho$ as 
\begin{eqnarray}\label{eq:rho2}
\rho = \diag(X(\rho)X(\rho)\zz).
\end{eqnarray}

Alternatively, the function \be \label{eq:cV}
V:= \cV(\rho) =  L^\dagger \rho + \mu_{xc}(\rho)\zz e
\ee
 is called potential and  
 the Hamiltonian matrix $\hat H(\rho)$, by convenient abuse of notation,  can  be expressed as 
\be \label{eq:H-V}  H(V):= \half L + V_{ion}  +
\Diag(V). \ee
 Obviously, it holds $\hat H(\rho)= H(V(\rho))$. Therefore, $X$ can be
 interpreted as an implicit function of
$V$. Let $X(V)\in \Rnk$ be the eigenvectors
corresponding to the $p$ smallest eigenvalues of $H(V)$. Then,
the fixed point map  \eqref{eq:rho2} is a system of  nonlinear equations with respect to $V$ as  
\be\label{eq:V}
  \left\{  \begin{aligned}
V &= \cV( F_\phi(V)),\\ 
F_\phi(V) &= \diag(X(V)X(V)\zz).
\end{aligned} \right.
\ee

The fixed point map \eqref{eq:V} is well defined if 
there is a gap between the $p$th and
$(p+1)$st smallest eigenvalues of $H(V)$. 
However, when these two eigenvalues are equal, there exists ambiguity on choosing the
eigenvectors $X(V)$ since the multiplicity is greater than one. A common
approach is to revise $F_\phi(V)$ in \eqref{eq:V}
by constructing a proper filter function. Let $q_1(V) , \ldots, q_n(V) $
be the eigenvectors of $H(V)$ associated with eigenvalues
$\lambda_1(V), \ldots, \lambda_n(V)$, respectively.  A particular choice of the filter function is the Fermi-Dirac distribution of the form
\be \label{eq:FDD} f_{\mu}(t) := \frac{1}{1+e^{\beta(t-\mu)}}, \ee
where $\mu$ is the solution of the equations
\be \label{eq:FDD-mu} \sum\limits_{i=1}^n f_\mu(\lambda_i(V)) = p. \ee
Since the left hand side of \eqref{eq:FDD-mu}
is monotonic with respect to $\mu$ for a fixed $\beta$,
the solution to \eqref{eq:FDD-mu} is unique for any choice of $\beta$ and $\lambda_i$.
Then the fixed map \eqref{eq:V} is replaced by the approximation
\be\label{eq:F-fmu1}
  \left\{  \begin{aligned}
V &= \cV( F_{f_\mu}(V)),\\ 
F_{f_\mu}(V) &= \diag\left(\sum\limits_{i=1}^n
 f_{\mu}(\lambda_i(V)) q_i(V) q_i(V)\zz\right).
\end{aligned} \right.
\ee

\subsection{The Jacobian of the Fixed Point Maps}\label{sec:Jac}
We first reformulate the functions $F_\phi(V)$ in \eqref{eq:V} and $F_{f_\mu}(V)$ in
\eqref{eq:F-fmu1} as the form of spectral operators. Using the differentiability of spectral
operators,  they can be proved to be differentiable  under some conditions.
Let $\{\lambda_i(V), q_i(V)\}$ be the eigenpairs of $H(V)$ and assume that the eigenvalues
$\lambda_1(V), \ldots, \lambda_n(V)$ are sorted in an ascending order,
\[ \lambda_1(V) \le \ldots \le \lambda_p(V) \le \lambda_{p+1}(V) \le \ldots \le
\lambda_n(V).\]
The eigenvalue decomposition of $H(V)$ can be written as
\be\label{eq:H-eig}
H(V) = Q(V) \Pi(V) Q(V)\zz,
\ee
where  $Q(V)$ and  $\Pi(V)$ are
\begin{equation} \label{eq:Q}
Q(V) = [q_1(V), \; q_2(V), \; \ldots, q_n(V)] \in \Rnn \; \mbox{ and  }  \;
\Pi(V) = \Diag(\lambda_1(V),\lambda_2(V), \ldots, \lambda_n(V)) \in
\Rnn.
\end{equation}
Hence, the function $F_\phi(V)$ in \eqref{eq:V} is equivalent to 
\be\label{eq:F-phi}
  F_{\phi}(V) = \ \diag(Q(V) \phi(\Pi(V)) Q(V)\zz),
\ee
where $\phi(\Pi) = \Diag( \phi(\lambda_1(V)),\phi(\lambda_2(V)), \ldots,
\phi(\lambda_n(V)))$
and
\be \label{eq:phi-complete}
\phi(t) := \begin{cases} 1 & \mbox{ for }
t \leq \frac{\lambda_p(V)+\lambda_{p+1}(V)}{2}, \\
  0 & \mbox{ for } t > \frac{\lambda_p(V)+\lambda_{p+1}(V)}{2}.
  \end{cases}
\ee
Similarly, the function $F_{f_\mu}(V)$ in \eqref{eq:F-fmu1} in the spectral
operator form is 
\be\label{eq:F-fmu}
 F_{f_\mu}(V) =  \diag(Q(V) f_{\mu}(\Pi(V)) Q(V)\zz).
\ee


Let  $\mu_1,\cdots,\mu_{r(V)}$ be the distinct eigenvalues among
$\{\lb_1(V),\cdots,\lb_n(V)\}$,  $r(V)$ be  the total number of
distinct values and $r_p(V)$ be the number of distinct
eigenvalues no greater than $\lb_p$. For any $k=1,\cdots,r(V)$, the set of
indices $i$ such that $\lambda_i=\mu_k$ is denoted by
 $\alpha_k:=\{i\mid \lb_i=\mu_k,\; i=1,\cdots,n\} $.
The next lemma shows the directional derivative of $F_{\phi}(V)$
by using  the differentiability of the spectral operators
\cite{ChenQiTseng2003,Ding2012,Lewis2001,Torki2001,Shapiro2002}.

\begin{lemma}\label{lemma:Jacobian-HV-phi} 
Suppose that 
Assumption \ref{asmp:UWP} holds at $H(V)$,
i.e., $\lambda_{p+1}(V)>\lambda_p(V)$.
Then $F_{\phi}(V)$ is continuously differentiable and its directional derivative  at $V$ along $z\in R^n$ is 
\begin{eqnarray}\label{eq:Jacobian-VF-phi-F}
\partial_V F_{\phi}(V)[z] =  \diag\left(
Q(V)\left(g_{\phi}(\Pi(V))\circ \left(Q(V)\zz\Diag\left(  z\right)
 Q(V)\right)\right)Q(V)\zz\right),
\end{eqnarray}
where ``$\circ$" denotes the Hadamard product between two matrices,
and $g_{\phi}(\Pi(V))\in \Rnn$ is the so-called first divided
difference matrix defined as
\begin{eqnarray}\label{eq:gdef}
(g_{\phi}(\Pi(V)))_{ij}=\left\{
\begin{array}{cc}
\frac{1}{\lb_i(V)-\lb_j(V)} & \mbox{ if }
 i\in\alpha_k,\,j\in\alpha_l,\,k\leq r_p(V),\, l>r_p(V),\\
\frac{-1}{\lb_i(V)-\lb_j(V)} & \mbox{ if }
i\in\alpha_k,\,j\in\alpha_l,\,k> r_p(V),\, l \leq r_p(V),\\
0 & \mbox{otherwise.}
\end{array}
\right.
\end{eqnarray}
\end{lemma}
\begin{proof}
The chain rule gives
\begin{eqnarray}\label{eq:Jacobian-F2}
\partial_V  F_{\phi}(V)[z] &=& 
 \frac{d \diag\left(Q\phi(\Pi)Q\zz\right)}{d  H}
\left[\partial_{V}  H(V)[z]\right].
\end{eqnarray}
 By applying
 the continuous differentiability of the spectral operators
in Proposition 2.10 of \cite{Ding2012}, the function $Q\phi(\Pi)Q\zz$
is differentiable with respect to $H$ and its directional derivative is given by
\begin{eqnarray}\label{eq:Jacobian-H1}
\frac{d Q\phi(\Pi)Q\zz}{d H} [S] =
Q\left(g_{\phi}(\Pi)\circ
\left(Q\zz S Q\right)\right)Q\zz, \quad \mbox{ for all }\, S\in\mathbb{S}^n,
\end{eqnarray}
where, for any $i,j=1,...,n$,
\begin{eqnarray}\label{eq:gdef-org}
(g_{\phi}(\Pi(V)))_{ij}=
\begin{cases}
\frac{\phi(\lb_i(V))-\phi(\lb_j(V))}{\lb_i(V)-\lb_j(V)} & \mbox{if\;} i\in\alpha_k,\,j\in\alpha_l,\,k\neq l,\\
0 & \mbox{otherwise.}
\end{cases}
\end{eqnarray}
Substituting \eqref{eq:phi-complete} into \eqref{eq:gdef-org} yields the specific
form of $g_{\phi}(\pi(V))$ in \eqref{eq:gdef}.
Since $\diag(\cdot)$ is a linear function, we obtain  
\begin{eqnarray}\label{eq:Jacobian-H2}
\frac{d \diag\left(Q\phi(\Lb)Q\zz\right)}{d H} [S]
&=& \frac{d \diag\left(Q\phi(\Lb)Q\zz\right)}{d Q\phi(\Lb)Q\zz}\frac{d Q\phi(\Lb)Q\zz}{d H}[S]\nonumber\\
&=& \diag\left(
Q\left(g_{\phi}(\Pi)\circ
\left(Q\zz S Q\right)\right)Q\zz\right), \quad \mbox{ for all }\,
S\in\mathbb{S}^n.
\end{eqnarray}
It follows from \eqref{eq:H-V} that
\begin{eqnarray}\label{eq:H-Jacobian}
\partial_V  H(V)[z] = \Diag(z).
\end{eqnarray}
Plugging \eqref{eq:Jacobian-H2} and \eqref{eq:H-Jacobian} into
\eqref{eq:Jacobian-F2}, we obtain \eqref{eq:Jacobian-VF-phi-F}.  This completes the proof.
\end{proof}

\begin{remark} Computing $\partial_V F_{\phi}(V)[z]$ requires all the
  eigenvectors $Q(V)$ and all eigenvalues $\Pi(V)$.
Let $E_{j,p}$ ($O_{j,p}$) be the $j\times p$ matrix with ones (zeros) at all its
entries. Then the matrix $g_{\phi}(\Pi(V))\in \Rnn$ takes the specific form
\[ g_{\phi}(\Pi(V)) =
\left(
\begin{array}{cc}
O_{p,p} & G\\
G\zz & O_{n-p,n-p}\\
\end{array}
\right), \]
where
\[ G =
\left(
\begin{array}{ccc}
\frac{1}{\mu_{1}-\mu_{r_p(V)+1}}
E_{|\alpha_{1}|,|\alpha_{r_p(V)+1}|}
& \cdots &
\frac{1}{\mu_{1}-\mu_{r(V)}}
E_{|\alpha_{1}|,|\alpha_{r(V)}|}\\
\vdots & \ddots & \vdots \\
\frac{1}{\mu_{r_p(V)}-\mu_{r_p(V)+1}}
E_{|\alpha_{r_p(V)}|,|\alpha_{r_p(V)+1}|}
& \cdots &
\frac{1}{\mu_{r_p(V)}-\mu_{r(V)}}
E_{|\alpha_{r_p(V)}|,|\alpha_{r(V)}|}\\
\end{array}
\right).
\]
\end{remark}

The directional derivative of  $F_{f_\mu}(V)[z]$ can be assembled in a similar fashion.
\begin{lemma}\label{lemma:Jacobian-F-fmu-rho}
The function $F_{f_\mu}(V)$ is continuously differentiable and its directional derivative at $V$ along $z\in R^n$ is
 \begin{eqnarray}\label{eq:Jacobian-VF-fmu-F}
\partial_V F_{f_\mu}(V)[z]  = \diag\left(Q(V)\left(g_{f_\mu}(\Pi(V))\circ
\left(Q(V)\zz \Diag\left(  z\right) Q(V)\right)\right)  Q(V)\zz\right),
\end{eqnarray}
where $g_{f_\mu}(\Pi(V))\in \Rnn$ is defined as, for any $i,j=1,...,n$,
\begin{eqnarray}\label{eq:fmudef}
(g_{f_\mu}(\Pi(V)))_{ij}=
\begin{cases}
\frac{f_\mu(\lb_i(V))-f_\mu(\lb_j(V))}{\lb_i(V)-\lb_j(V)} & \mbox{if\;} i\in\alpha_k,\,j\in\alpha_l,\,k\neq l,\\
f_\mu'(\lb_i(V)) & \mbox{otherwise.}
\end{cases}
\end{eqnarray}
\end{lemma}

We next compute the Jacobian
 of $\cV(F_{\phi}(V))$ and $\cV(F_{f_\mu}(V))$.


\begin{theorem}\label{thm:Jacobian-HV-phi}
Let $J(\rho)$ be defined as \eqref{eq:def-J-rho}.
\begin{enumerate}
  \item 
 Suppose that 
  Assumption \ref{asmp:UWP} holds at $H(V)$,
    i.e., $\lambda_{p+1}(V)>\lambda_p(V)$.
Then 
the
  Jacobian of
$\cV(F_{\phi}(V))$ at $V$ is 
\begin{eqnarray}\label{eq:Jacobian-VF-phi}
\partial_V \cV(F_{\phi}(V))[z] =J(F_{\phi}(V))\partial_V F_{\phi}(V)[z],\quad \mbox{ for all }\,z\in\mR^n.
\end{eqnarray}
\item 
The 
 Jacobian of
$\cV(F_{f_\mu}(V))$ at $V$ is 
\begin{eqnarray}\label{eq:Jacobian-VF-fmu}
\partial_V \cV(F_{f_\mu}(V))[z] =J(F_{f_\mu}(V)) \partial_V F_{f_\mu}(V)[z],\quad \mbox{ for all }\,z\in\mR^n.
\end{eqnarray}
\end{enumerate}
\end{theorem}
\begin{proof}
  Note that 
  \begin{eqnarray}\label{eq:V-Jacobian}
\partial_\rho(\cV(\rho))[z] = J(\rho) z,\quad \mbox{ for all }\, z\in\mR^{n}.
\end{eqnarray}
Applying the chain
rules to $\partial_V  \cV(F_{\phi}(V))[z]$ and using \eqref{eq:V-Jacobian} and \eqref{eq:Jacobian-VF-phi-F}, we
obtain \eqref{eq:Jacobian-VF-phi}. This completes the proof.
\end{proof}

\section{Convergence of the SCF iteration}\label{sec:conv}

\subsection{The SCF Iteration and the Simple Mixing Scheme}
Starting from an initial vector $V^0\in\R^n$, the SCF iteration for solving the
fixed point map \eqref{eq:V}
recursively computes the eigenpairs $\{X(V^{i+1}), \Lambda(V^{i+1})\}$ as the solution of the linear eigenvalue
problem:
\[
\begin{aligned}   H(V^i) X(V^{i+1}) &=  X(V^{i+1}) \Lambda(V^{i+1}), \\
  X(V^{i+1})\zz X(V^{i+1}) &=  I,
  \end{aligned} \]
  and then the potential is updated as
  \be\label{eq:V-SCF}V^{i+1} = \cV(F_\phi(V^i)). 
  \ee
  When the difference between $V^i$ and $V^{i+1}$ is negligible, the
system is said to be self-consistent and the SCF iteration is terminated.

The SCF iteration often converges slowly or even fails to
converge. One of the heuristics for accelerating and stabilizing the SCF
iteration is charge or potential mixing  \cite{Kerker1981,Kresse-1996}.
Basically, the new potential $V^{i+1}$ is constructed from a linear combination of the previously
computed potential and the one obtained from certain
 schemes at current iteration.  In particular, the simple mixing scheme replaces
 \eqref{eq:V-SCF} by updating
 \begin{eqnarray}\label{eq:V-simple-mixing}
 V^{i+1} = V^i - \alpha (V^i - \cV(F_\phi(V^i))),
 \end{eqnarray}
  where $\alpha$ is a properly chosen step size.
 Similarly, the SCF iteration using simple mixing for solving the fixed point map
 \eqref{eq:F-fmu1}
  is
 \begin{eqnarray}\label{eq:V-simple-mixing-fmu}
 V^{i+1} = V^i - \alpha (V^i - \cV(F_{f_\mu}(V^i))).
 \end{eqnarray}

\subsection{Global Convergence Analysis} 
We first make the following assumptions.
\begin{assumption}\label{asmp:exc}
The second-order  derivatives of the exchange correlation functional
 $\epsilon_{xc}(\rho)$ is uniformly bounded from above. 
 Without loss of generality,
 we assume that there exists a constant $\theta$ such that
 \begin{eqnarray}\label{eq:exc}
 \left\| \partial  \mu_{xc}(\rho) e \right\|_2 \le \theta
 ,\quad
 \mbox{ for all }\,\rho \in\mR^{n} .
 \end{eqnarray}
\end{assumption}

  Although we cannot verify Assumption \ref{asmp:exc} for any $X\in \Rnk$, it
  holds at a strong local minimizer  using our
  lower bounds for nonzero charge densities in subsection  \ref{sec:low-bound}
  if the exchange correlation energy  is \eqref{eq:exdef}.

%

It can be verified from the definition of the operator  $\partial_V F_{\phi}(V)[\cdot]$ in
\eqref{eq:Jacobian-VF-phi-F} that it is a linear map.  The induced $\ell_2$-norm of $\partial_V\cV(F_{\phi}(V))$ and $\partial_V F_{\phi}(V)[\cdot]$ are defined as 
\be \label{eq:J-2-norm}\|\partial_V\cV(F_{\phi}(V))\|_2 = \max_{z\neq 0} \frac{\|\partial_V\cV(F_{\phi}(V))[z]
\|_2 }{\|z\|_2 } \mbox{ and }   \|\partial_V F_{\phi}(V)\|_2 = \max_{z\neq 0} \frac{\|\partial_V F_{\phi}(V)[z]
\|_2 }{\|z\|_2 },\ee
respectively.
The next lemma shows that their $\ell_2$-norms are bounded  if Assumption \ref{asmp:UWP} holds at $H(V)$.
\begin{lemma}\label{lm:J-bound}
 If Assumption \ref{asmp:UWP} holds at $H(V)$ for a given $V\in\mathbb{R}^n$, then it holds 
\begin{eqnarray}\label{eq:Jacobianbound-V}
\|\partial_V F_{\phi}(V)\|_2 \leq \frac{1}{\delta} \quad \mbox{ and } \quad \|\partial_V\cV(F_{\phi}(V))\|_2 \le  \frac{\|L^\dagger\|_2+ \theta}{\delta}.
\end{eqnarray}
\end{lemma}
\begin{proof}
For any $z\in\mR^n$,  we obtain
\begin{eqnarray}
  \label{eq:2norm}
\|\partial_V F_{\phi}(V)[z]\|_2  &=& \| \diag\left(
Q(V)\left(g_{\phi}( \Pi(V))\circ \left( Q(V)\zz\Diag\left(  z\right)
 Q(V)\right)\right) Q(V)\zz\right)\|_2 \nonumber\\
&\leq& \|Q(\rho)\left(g_{\phi}(\Pi(\rho))\circ
\left(Q(\rho)\zz \Diag(z) Q(\rho)\right)\right)Q(\rho)\zz\|\ff\nonumber\\
&=&  \|g_{\phi}(\Pi(\rho))\circ
\left(Q(\rho)\zz\Diag(z)  Q(\rho)\right)\|\ff\nonumber\\
&\leq& \frac{1}{\delta}  \|Q(\rho)\zz \Diag(z) Q(\rho)\|\ff \nonumber\\
&\leq & \frac{1}{\delta}  \| z \|_2,
\end{eqnarray}
where the second inequality is due to $|(g_{\phi}(\Pi(\rho)))_{ij}| \leq 1/\delta$.
Then the first inequality in \eqref{eq:Jacobianbound-V} holds from the
definitions \eqref{eq:J-2-norm} and \eqref{eq:2norm}. 
It follows from \eqref{eq:Jacobian-VF-phi} and \eqref{eq:2norm} 
that 
\begin{eqnarray}\label{eq:Jacobianbound}
\|\partial_V\cV(F_{\phi}(V))[z]\|_2 
\leq \|J(F_\phi(V))\|_2 \|\partial_V F_{\phi}(V)[z]\|_2
\leq \frac{\|L^\dagger\|_2+ \theta}{\delta}\|z\|_2.
\end{eqnarray}
This completes the proof.
\end{proof}

 The set $\{H(V)\mid V\in\mR^n\}$ is called uniformly well posed (UWP) \cite{LeBris2005,YangGaoMeza2009} 
with respect to a constant $\delta>0$ if Assumption \ref{asmp:UWP} holds at $H(V)$ with $\delta$ for any $V\in\mR^n$. We next establish the convergence
of the simple mixing scheme \eqref{eq:V-simple-mixing} when UWP holds. 
\begin{theorem}\label{thm:glconv-SCF}
Suppose that Assumption \ref{asmp:exc} holds and $\{H(V)\mid V\in\mR^n\}$ is UWP
with a  constant $\delta$ such that
\begin{eqnarray}\label{eq:gl-conv-cond}
b_1:= 1 - \frac{\|L^\dagger\|_2+ \theta}{\delta} > 0.
\end{eqnarray}
Let
$\{V^i\}$ be a sequence generated
by the simple mixing scheme \eqref{eq:V-simple-mixing} using
a step size $\alpha$ satisfying 
\begin{eqnarray}\label{eq:gl-conv-alpha}
0<\alpha<\frac{2}{2-b_1}.
\end{eqnarray}
Then $\{V^i\}$ converges to a solution of the KS equation \eqref{eq:KS}
with linear convergence rate no more than $|1-\alpha| + \alpha (1-b_1) $.
\end{theorem}
\begin{proof}
For any $V^i$, it follows from \eqref{eq:Jacobianbound}, \eqref{eq:gl-conv-cond} and \eqref{eq:gl-conv-alpha} 
that
\beaa
&&\|(1-\alpha)I + \alpha \partial_V\cV(F_{\phi}(V^i))\|_2 \nonumber\\
&\leq& |1-\alpha| + |\alpha| \|\partial_V\cV(F_{\phi}(V^i))\|_2 \nonumber\\
&\leq& \left\{\begin{array}{ll}
1-\alpha + \alpha \frac{\|L^\dagger\|_2+ \theta}{\delta} \,=\, 1-\alpha b_1,
& \mbox{ if } 0<\alpha <1\\
\alpha - 1 + \alpha \frac{\|L^\dagger\|_2+ \theta}{\delta} \,=\, \alpha (2-b_1)-1,
& \mbox{ if } \alpha \ge 1
\end{array} \right. \\
&<& 1,
\eeaa
which completes the proof.
\end{proof}

\begin{remark}\label{rmk:1}
When the step size $\alpha = 1$, the simple mixing scheme \eqref{eq:V-simple-mixing} becomes
the SCF iteration \eqref{eq:V-SCF} with the convergence rate
$\frac{\|L^\dagger\|_2+ \theta}{\delta}$.  Since neither $p$ nor $n$ is involved
in \eqref{eq:gl-conv-cond}, it is much weaker than 
$\frac{12k\sqrt{n}\|L^\dagger\|_2+ \theta}{\delta} < 1$
 required by Theorem 1 in \cite{LiuWangWenYuan}.
\end{remark}

We next establish  convergence to the solutions of
the modified fixed-point map  \eqref{eq:F-fmu1}
  without assuming the UWP properties.
\begin{theorem}\label{thm:glconv-SCF-fmu}
Suppose that Assumption \ref{asmp:exc} holds and 
\begin{eqnarray}\label{eq:gl-conv-cond-fmu}
b_2:=  1 - \frac{\beta(\|L^\dagger\|_2+ \theta)}{4} > 0.
\end{eqnarray}
Let $\{V^i\}$ be a sequence generated
by the simple mixing scheme \eqref{eq:V-simple-mixing-fmu}  using a step size
$\alpha$ satisfying 
\begin{eqnarray}\label{eq:gl-conv-alpha2}
0<\alpha<\frac{2}{2-b_2}.
\end{eqnarray}
Then  the sequence $\{V^i\}$ converges to a solution
of  \eqref{eq:F-fmu1}
with linear convergence rate no less than  $|1-\alpha|+\alpha(1-b_2)$.
\end{theorem}
\begin{proof}
Using the mean value theorem and the fact that
\[ |f_{\mu}'(t)| = \left|\frac{-\beta e^{\beta(t-\mu)}}{(1+e^{\beta(t-\mu)})^2}\right|
\leq \frac{\beta}{4}, \]
we obtain $|(g_{f_{\mu}}(\Pi(V)))_{ij}| \leq \beta/4$, which yields 
\[\|\partial_V \cV(F_{f_\mu}(V))\|_2 \leq \frac{\beta(\|L^\dagger\|_2+
\theta)}{4}.\]
Then, the convergence of \eqref{eq:V-simple-mixing-fmu} is proved similar to that
of Theorem \ref{thm:glconv-SCF}.
\end{proof}

\begin{remark}
Suppose that 
UWP holds and $f_{\mu}$ is chosen such that
\begin{eqnarray}\label{bond}
\begin{cases}
\frac{1}{1+e^{\beta(\lambda_p-\mu)}}\geq 1-\gamma,\\
\frac{1}{1+e^{\beta(\lambda_{p+1}-\mu)}}\leq \gamma,
\end{cases}
\end{eqnarray}
where $\gamma \ll 1$ is a constant.
It can be shown that
$\beta\geq \frac{2}{\delta}\cdot\ln\frac{1-\gamma}{\gamma}$.
Hence, we have $\frac{\beta}{4} \geq \frac{1}{\delta}$ and the condition
\eqref{eq:gl-conv-cond} is implied by \eqref{eq:gl-conv-cond-fmu} when $\ln\frac{1-\gamma}{\gamma}\geq 2$  or equivalently $\gamma\leq \frac{1}{e^2+1}
\approx 0.12$.
On the other hand, the closer $\gamma$
is to zero, the closer $f_{\mu}$ is to $\phi$ from \eqref{bond}.
 Therefore, the convergence rate of the fixed-point iteration using
 $F_{\phi}$ is better than that
 of $F_{f_\mu}$ when $F_{f_\mu}$ is sufficiently close to $F_{\phi}$.
\end{remark}

\begin{remark}
The convergence of the SCF iteration without simple
mixing for solving a special KS equation without the exchange correlation energy is
established in \cite{YangGaoMeza2009} under the condition
\begin{eqnarray} \label{eq:conv-yang-gao-meza}
\frac{n^4\beta\|L^\dagger\|_2}{2}<1.
\end{eqnarray}
We can see that our condition is weaker than \eqref{eq:conv-yang-gao-meza} since $n^4$ is not required. 
\end{remark}

\subsection{Local Convergence Analysis}
Suppose that $V^*$ is a solution of the fixed point map
\eqref{eq:V-simple-mixing}. Let $B(V^*,\eta):= \{V\mid ||V-V^*||_2\leq
\eta\}$ be a neighborhood
of $V^*$ for a given $\eta>0$. The Taylor expansion at $V^*$ yields 
\bea\label{eq:aaa}
V^{k+1} - V^* &=& V^k -  \alpha (V^k - \cV(F_\phi(V^k))) 
- (V^* -  \alpha (V^* - \cV(F_\phi(V^*))))\nonumber\\
&=& (I - \alpha (I - \partial_V \cV(F_\phi(V^*)))) [V^k - V^*]
+ o(||V^k - V^*||_2),\quad \mbox{ for all } \,V^k\in B(V^*,\eta).
\eea
If the spectral radius of the operator $I - \alpha (I - \partial_V \cV(F_\phi(V^*)))$
is less than one, there must exist a sufficiently small $\eta$ so that 
the simple mixing scheme \eqref{eq:V-simple-mixing} initiating from a point in
$B(V^*,\eta)$ converges to $V^*$ linearly. 

We first present a few properties of the linear operators. Denote the space of linear operators  by 
\[ \LL(\mR^n,\mR^n):=\{\cP\mid
\cP:\mR^n \mapsto \mR^n \mbox{ is\ a\ linear\ map}\}.\]   Since $\LL(\mR^n,\mR^n)$ is isomorphic
to $\mR^{n\times n}$,  the eigenvalue, eigenvector and the spectrum
for any linear operator can be defined similar to a matrix.
For a given $\cP\in \LL(\mR^n,\mR^n)$, if a scalar $\lambda\in\C$ and a nonzero vector $z \in \C^n$ satisfy
\be\label{eq:opdef} \cP[z]  = \lambda z,\ee
the scalar $\lambda$ and the vector $z$ are called the eigenvalue and  
eigenvector of $\cP$, respectively. The spectrum of $\cP$, denoted by
$\lambda(\cP)$, is the set consisting of all the eigenvalues of $\cP$. The spectral
radius, denoted
by $\varrho(\cP)$, is the largest absolute value of all elements in its spectrum.
 The operator $\cP$ is called symmetric if $y\zz \cP[x] = x\zz
\cP[y]$ for any $x,y\in \R^n$.

\begin{definition}\label{def:op-mat}
Given $\cP\in \LL(\mR^n,\mR^n)$, the matrix $P=(\cP[e_1], \ldots,
\cP[e_n])$  is called the basic transformation matrix of $\cP$, where $e_i$,
$i=1,2,...,n$, is the $i$th column of the identity matrix. A linear operator $\cP^* \in \LL(\mR^n,\mR^n)$
is called the adjoint operator of $\cP$ if $\cP^*[x] =
P\zz x$ holds for all $x\in \R^n$. 
\end{definition}

%


Let $P$ be the basic transformation matrix of $\cP\in \LL(\mR^n,\mR^n)$.
Then $P$ is symmetric if and only if $\cP$ is symmetric. Moreover, $\cP$ and $P$ has the same spectrum since $ \cP[z] =  P z$.
 Let $M_1,M_2\in \mR^{n\times n}$ be two real matrices and $P_1$ and $P_2$ be the basic transformation matrices of $\cP_1, \cP_2\in \LL(\mR^n,\mR^n)$, respectively.
Then $M_1 P_1 + M_2 P_2$ is the basic transformation matrix of the linear operator $M_1
\cP_1 + M_2 \cP_2$.  
A linear operator $\cP\in\LL(\mR^n,\mR^n)$ is called positive semidefinite if
 $z\zz \cP[z] \ge 0$ for all  $z\in\R^n$.  We next show that the eigenvalues of the
product of a symmetric matrix and a symmetric positive semidefinite linear
operator are real.
\begin{lemma}\label{cl:op-prop2}
Suppose that $M\in\mR^{n\times n}$ is a symmetric matrix,
and $\cP\in\LL(\mR^n,\mR^n)$ is a symmetric positive semidefinite linear operator.
Then all the eigenvalues of the linear operator $M\cP$ are real. Furthermore, it
holds
\begin{eqnarray}
\label{eq:eig-op-prod0}
\lambda_{\max}(M\cP) &\leq& \begin{cases} \lambda_{\max}(M)\lambda_{\max}(\cP),
&\quad \mbox{if\,\,}\lmax(M)\geq 0,\\
\lambda_{\max}(M)\lambda_{\min}(\cP),
&\quad \mbox{otherwise}, \end{cases}\\
\label{eq:eig-op-prod2}
\lambda_{\min} (M\cP) &\ge& \begin{cases} \lambda_{\min}(M)\lambda_{\min}(\cP) ,
&\quad \mbox{if\,\,}\lmin(M)\geq 0,\\
\lambda_{\min}(M)\lambda_{\max}(\cP),
&\quad \mbox{otherwise}.\end{cases}
\end{eqnarray}    
\end{lemma}
\begin{proof} Let $P$ be the basic transformation matrix of $\cP$. It suffices 
  to prove the statements with $\cP$ replaced by $P$. 
Since $\cP$ is symmetric positive semidefinite, $P$ is also symmetric positive
semidefinite. Hence, it can be diagonalized as $P=U D U^T$, where $U$ is orthogonal and
$D=\Diag(\mu_1,\ldots,\mu_n)$ such that $\mu_i\ge 0$.
Define $D^\half:=\Diag(\mu_1^\half,\ldots,\mu_n^\half)$ and write
$P^\half=U D^\half U^T$. Then we obtain $P=P^\half P^\half$.  We now prove that every eigenvalue of $R := P^\half M P^\half$
is an eigenvalue of $MP$ and vice versa.
It is known that the eigenvalues of a matrix are continuous functions of the
matrix entries.  Let $D_\epsilon:=D+\epsilon I$ 
and $P_\epsilon:= U D_\epsilon U^T$ for $\epsilon \ge 0$. Then $P_\epsilon\to P$ and
$P_\epsilon^\half:=U D_\epsilon^\half U^T\to P^\half$ as $\epsilon\to 0$.
 Hence, $MP_\epsilon \to MP $ and
$R_\epsilon:= P_\epsilon^\half M P_\epsilon^\half\to R$ as $\epsilon\to 0$.
Since $P_\epsilon^\half$ is invertible, we have
$R_\epsilon= P_\epsilon^{\half} M P_\epsilon  P_\epsilon^{-\half} $. 
Therefore, $R_\epsilon$ and $M P_\epsilon$  have the
same eigenvalues. As $\epsilon\to 0$, these eigenvalues converge
to those of $R$ and $MP$, respectively. Hence, $R$ and $MP$
have the same eigenvalues. The symmetry of $R$ further implies that the eigenvalues of
$MP$ are real. 

Since $\lmax(M) I \succeq M $, 
we obtain 
\beaa
\lmax(M)P = \Ph (\lmax(M) - M) \Ph + \Ph M \Ph \succeq  \Ph M \Ph,
\eeaa
which yields \eqref{eq:eig-op-prod0} since the eigenvalues of $R = P^\half M
P^\half$ and $MP$ are the same.
Similarly, \eqref{eq:eig-op-prod2} holds due to $M \succeq \lmin(M) I$ and 
\beaa
 \Ph M \Ph = \Ph (M - \lmin(M) ) \Ph + \lmin(M)P \succeq  \lmin(M)P.
\eeaa
This completes the proof.
\end{proof}

The next lemma shows that $\partial_V F_{\phi}(V)[\cdot]$ is negative semidefinite. 
\begin{lemma} \label{lm:J-neg}
For any  $z\in\R^n$, it holds $z\zz \partial_V F_{\phi}(V)[z]\le 0$.
\end{lemma}
\begin{proof}
For any $z\in R^n$, we have 
\beaa z\zz \partial_V F_{\phi}(V)[z]  &=& z\zz \diag\left(Q(V)\left(g_\phi(\Pi(V))\circ
\left(Q(V)\zz \Diag\left(  z\right) Q(V)\right)\right)  Q(V)\zz\right) \nonumber \\
&=&  \iprod{\left(Q(V)\zz \Diag\left(  z\right) Q(V)\right)}{g_\phi(\Pi(V))\circ
\left(Q(V)\zz \Diag\left(  z\right) Q(V)\right) } \nonumber\\
&=& e\zz \left(g_\phi(\Pi(V))\circ
\left(Q(V)\zz \Diag\left(  z\right) Q(V)\right) \circ
\left(Q(V)\zz \Diag\left(  z\right) Q(V)\right) \right) e \\
&\le& 0,
\eeaa
where the third equality uses the properties of the Hadamard products and the
inequality is due to 
\[ \left(Q(V)\zz \Diag\left(  z\right) Q(V)\right) \circ
\left(Q(V)\zz \Diag\left(  z\right) Q(V)\right) \ge 0 \mbox{ and }
g_\phi(\Pi(V))\le 0. \] 
This completes the proof.
\end{proof}

%

We now establish the local convergence result for the
simple mixing scheme.
\begin{theorem}\label{thm:conv-SCF}
Let $V^*$ be a solution of the KS equation \eqref{eq:KS}.  Suppose that Assumption \ref{asmp:exc} holds and Assumption \ref{asmp:UWP} is
valid at $H(V^*)$ with a constant $\delta$ satisfying  
\be \label{eq:conv-cond}  \delta >
-\lambda_{\min}^*,\ee where 
$ \lambda_{\min}^* := \min\{0, \lambda_{\min}(J(F_\phi(V^*)))\}$. 
There exists an open neighborhood $\Omega$ of $V^*$, such that 
the sequence
$\{V^i\}$ generated
by the simple mixing scheme \eqref{eq:V-simple-mixing} using $V^0\in \Omega$ and a step size  
\begin{eqnarray}\label{eq:alpha}
\alpha \in \left(0,\frac{2\delta}{||L^\dagger||_2 + \theta + \delta}\right)
\end{eqnarray}
converges to $V^*$
with R-linear convergence rate no more than 
\[\max\left\{
\left(1 - \alpha \frac{\delta + \lambda_{\min}^*}{\delta} \right),
\left( \alpha \frac{||L^\dagger||_2 + \theta + \delta}{2\delta} -1 \right)
\right\}.
\]
\end{theorem}
\begin{proof}
The Taylor expansion \eqref{eq:aaa} implies that 
 local convergence of the scheme \eqref{eq:V-simple-mixing}
holds if 
\begin{eqnarray}\label{eq:App-Newton-Conv}
\varrho(I - \alpha  \cA  ) <1,
\end{eqnarray}
where $\cA:= I - J(\FFVS) \partial_V F_{\phi}(V^*)$.
According to Lemma \ref{lm:J-neg}, $-\partial_V F_{\phi}(V^*)$ is symmetric positive semidefinite.
Using Lemma \ref{cl:op-prop2}, we conclude that all the eigenvalues of
$\cA$ are real. Hence,  \eqref{eq:App-Newton-Conv} is guaranteed if
\begin{eqnarray}\label{eq:neg-lmin}
\lmin(\cA)  &>& 0;\\
\label{eq:neg-lmax}
\alpha\lmax(\cA) &<& 2.
\end{eqnarray}

 Note that $\lmin(\cA) = 1 + \lmin (J(\FFVS) (-\partial_V F_{\phi}(V^*)))$. Using Lemma \ref{cl:op-prop2}, $\lmax(-\partial_V F_{\phi}(V^*)) \le  \frac{1}{\delta}$ from Lemma \ref{lm:J-bound} and the definition of $\lmin^*$, we obtain \beaa \lmin(\cA) -1 &\ge& \begin{cases}\lmin(J(\FFVS)) \lmin(-\partial_V F_{\phi}(V^*)), & \mbox{ if } \lmin(J(\FFVS))\geq 0, \\
  \lmin(J(\FFVS)) \lmax(-\partial_V F_{\phi}(V^*)) , & \mbox{ otherwise}\end{cases} \\
 &\ge& \begin{cases} 0, & \mbox{ if
  } \lmin(J(\FFVS))\geq 0, \\
\frac{1}{\delta}\lmin(J(\FFVS)), & \mbox{ otherwise}\end{cases} \\
&\ge& \frac{\lmin^*}{\delta},
  \eeaa
  which yields \eqref{eq:neg-lmin} from the assumption \eqref{eq:conv-cond}.

Using Lemma \ref{cl:op-prop2} again, we have
\begin{eqnarray}\label{eq:lmaxM1}
\lmax(\cA) &\leq& 1 + \lmax (J(\FFVS) (-\partial_V F_{\phi}(V^*))) \nonumber\\
& \leq& 1 + \max\{0,\lmax (J(\FFVS)) \lmax(-\partial_V F_{\phi}(V^*))\}
\leq 1 + \frac{||L^\dagger||_2 + \theta}{\delta},
\end{eqnarray}
 which together with \eqref{eq:alpha} gives \eqref{eq:neg-lmax}.
\end{proof}

The condition \eqref{eq:conv-cond} can be much weaker than  
$ \|L^\dagger\|_2+ \theta < \delta$ required in Theorem \ref{thm:glconv-SCF}. 


\begin{corollary}\label{lm:Brho-lb}
Suppose that Assumption \ref{asmp:exc} holds. Then the condition
\eqref{eq:conv-cond} holds if 
 \be \label{eq:cond-lmin}
\max(\theta -  \lambda_{\min}(L^\dagger), 0) < \delta.
\ee
\end{corollary}
\begin{proof}
It follows from  \eqref{eq:def-J-rho}
and Assumption \ref{asmp:exc} that 
\[
\lmin(J(F_\phi(V))) = \lmin(L^\dagger + \partial \mu_{xc}(F_\phi(V))e)
\geq \lmin(L^\dagger) + \lmin(\mu _{xc}(F_\phi(V))e) 
\geq \lmin(L^\dagger) - \theta.
\]
Hence, \eqref{eq:conv-cond} holds from the definition of $\lambda_{\min}^*$.
\end{proof}

In particular,   when  $J(\FFVS)$ is positive semidefinite, we have
$\lambda_{\min}^*=0$ and   \eqref{eq:conv-cond} is a direct
consequence of Assumption \ref{asmp:UWP}.
 \begin{corollary} \label{cor-B-psd}
Suppose that Assumptions  \ref{asmp:UWP} holds at $H(V^*)$ and 
$J(\FFVS)$ is positive semidefinite. Then the condition
\eqref{eq:conv-cond} holds. \end{corollary}

%

We can obtain the following local convergence result for the modified fixed-point map \eqref{eq:F-fmu1} 
 in the same manner as Theorem \ref{thm:glconv-SCF-fmu}.
\begin{corollary}\label{thm:conv-SCF-fmu}
Suppose that Assumption \ref{asmp:exc} holds and 
\be \label{eq:conv-cond-fmu}  \frac{4}{\beta} >
-\lambda_{\min}^*,\ee where 
$ \lambda_{\min}^* := \min\{0, \lambda_{\min}(J(F_\phi(V^*)))\}$. 
Let $V^*$ be a solution of the KS equation \eqref{eq:KS}.
There exists an open neighborhood $\Omega$ of $V^*$, such that 
the sequence
$\{V^i\}$ generated
by the simple mixing scheme \eqref{eq:V-simple-mixing-fmu} using $V^0\in \Omega$ and a step size  
\begin{eqnarray}\label{eq:alpha-fmu}
\alpha \in \left(0,\frac{8}{(||L^\dagger||_2 + \theta)\beta + 4}\right)
\end{eqnarray}
converges to $V^*$
with R-linear convergence rate no more than 
\[\max\left\{
\left(1 - \alpha\frac{\lambda_{\min}^*\beta + 4}{4}\right),
\left(\alpha\frac{(||L^\dagger||_2 + \theta)\beta + 4}{8} -1 \right)
\right\}.
\]
\end{corollary}

\section{Convergence Analysis of Approximate Newton Approaches}\label{sec:App-Newton}

 The generalized Jacobian  $\partial_V\cV(F(V))$
 in \eqref{eq:Jacobian-VF-phi}  suggests that Newton's
method for solving the fixed point map \eqref{eq:V} is 
\[ V^{i+1} = V^i - \alpha \left(I- J(F_\phi(V^i))\partial_V F_{\phi}(V^i)\right)^{-1}
\left(V^i - \cV\left(F_{\phi}(V^i)\right)\right),\]
where $\alpha$ is a step size.  Obviously, this method is not computationally practical for
solving the fixed-point maps due to the presence of all eigenvectors
and eigenvalues in $\partial_V F_{\phi}(V)[\cdot]$. 
In this section, we propose two approximate Newton approaches  in the form
\be\label{eq:V-App-Newton}
V^{i+1} = V^i - \alpha \left(I- D^i\right)^{-1} \left(V^i - \cV\left(F_{\phi}(V^i)\right)\right),
\ee
where $\alpha>0$ and $D^i \in \R^{n\times n}$ is a matrix for approximating the Jacobian $\partial_V
\cV(F(V^i))$. 



\begin{theorem}\label{thm:glconv-SCF-AN}
Suppose that Assumption \ref{asmp:exc} and UWP hold. 
Let
$\{V^i\}$ be a sequence generated
by \eqref{eq:V-App-Newton} using
$\{D^i\}$ and a step size $\alpha$ such that 
\[
0<\alpha<\frac{2}{b_2}, \quad 0 < \gamma_{\min}\leq \sigma_{\min}(I-D^i) \mbox{ and } \sigma_{\max}(I-D^i) \leq \gamma_{\max},
\]
where $b_2 := 1 + \frac{\|L^\dagger\|_2+ \theta}{\delta}$, and $\sigma_{\min}
$ and $\sigma_{\max}$ are the smallest and largest singular values of $I-D^i$, respectively.
 If $b_1:= 1 - \frac{\gamma_{\max}}{\gamma_{\min}}
\frac{\|L^\dagger\|_2+ \theta}{\delta} > 0$, then $\{V^i\}$ converges to a solution of the KS equation \eqref{eq:KS}
with linear convergence rate no more than $\max(1-\alpha
\gamma_{\max}\inv  b_1,\alpha 
\gamma_{\min}\inv b_2 - 1) $.
\end{theorem}
\begin{proof}
For any $V^i$, it follows from the definitions of $D^i$, $\alpha$ and $b_2$  
that
\beaa
&&\|I - \alpha(I - D^i)\inv (I - \partial_V\cV(F_{\phi}(V^i)))\|_2 \nonumber\\
&=&\|I-\alpha(I - D^i)  + \alpha (I - D^i)\inv\partial_V\cV(F_{\phi}(V^i)) \|_2  \nonumber\\
&\leq& \|I-\alpha(I - D^i)\|_2 + |\alpha|  \|(I - D^i)\inv J(F_\phi(V^i))J(V^i)\|_2 \nonumber\\
&\leq& \left\{\begin{array}{ll}
1-\alpha \gamma_{\max}\inv+ \alpha \gamma_{\min}\inv\frac{\|L^\dagger\|_2+ \theta}{\delta} \,=\, 1-\alpha \gamma_{\max}\inv b_1,
& \mbox{if\,} \alpha <\gamma_{\max};\\
\alpha \gamma_{\min}\inv - 1 + \alpha \gamma_{\min}\inv\frac{\|L^\dagger\|_2+ \theta}{\delta} \,=\, \alpha \gamma_{\min}\inv  b_2-1,
& \mbox{otherwise,}
\end{array} \right. \\
&<& 1.
\eeaa
This completes the proof.
\end{proof}

\subsection{Approximate Newton Method I}

 Our first approach replaces the operator $\partial_V F_{\phi}(V^i)[\cdot]$ by a diagonal
 matrix $\tau^i I$, where $\tau^i$ is a non-positive scalar. It is chosen to be non-positive since $\partial_V F_{\phi}(V^i)[\cdot]$ is
 negative semidefinite from Lemma \ref{lm:J-neg}.  Consequently, we set $D^i :=
 \tau^i J(\rho)$ 
and the scheme \eqref{eq:V-App-Newton}  becomes
\be\label{eq:V-App-Newton-iv}
V^{i+1} = V^i - \alpha \left(I - \tau^i J(F_\phi(V^i))\right)^{-1} 
\left(V^i - \cV\left(F_\phi(V^i)\right)\right).
\ee
The next theorem presents the local convergence analysis for the method \eqref{eq:V-App-Newton-iv}.
\begin{theorem}\label{thm:conv-SCF-AN}Let $V^*$ be a solution of the KS equation
  \eqref{eq:KS}.  
Suppose that Assumption \ref{asmp:exc} holds  with a constant $\theta$ and 
Assumption \ref{asmp:UWP} is
valid at $H(V^*)$ with a constant $\delta$ satisfying 
\begin{eqnarray}\label{eq:conv-cond-AN}
\delta > -\lambda_{\min}^*,
\end{eqnarray}
 where 
$ \lambda_{\min}^* := \min\{0, \lambda_{\min}(J(F_\phi(V^*)))\}$. 
 Let $\{V^i\}$ be a
sequence generated
by the scheme \eqref{eq:V-App-Newton-iv} using  $\lim\limits_{i\rightarrow \infty}\tau^i
=\tau^*\in\left(-\frac{1}{\delta},0\right)$ and a step size  
\begin{eqnarray}\label{eq:alpha-AN2}
\alpha \in \left(0,\frac{\delta + \lmin^*}{||L^\dagger||_2 +
\theta+\delta}\right).
\end{eqnarray}
If the initial point $V^0$ is selected in a sufficiently small open neighborhood
 of $V^*$, then $\{V^i\}$  converges to $V^*$
with R-linear convergence rate no more than 
\[\max\left\{
\left(1 - \alpha \left(\frac{\delta}{||L^\dagger||_2 + \theta + \delta}
+ \frac{\lmin^*}{\delta + \lmin^*}\right) \right),
\left( \alpha \frac{||L^\dagger||_2 + \theta+\delta}{\delta + \lmin^*} -1 \right)
\right\}.
\]
\end{theorem}
\begin{proof}
The convergence of the iteration \eqref{eq:V-App-Newton-iv}
 is guaranteed by 
\begin{eqnarray}\label{eq:App-Newton-Conv-AN}
\varrho(I - \alpha \cM ) <1,
\end{eqnarray}
where $\cM = (I - \tau^* J(\FFVS))\inv (I - J(\FFVS) \partial_V F_{\phi}(V^*))$. A direct
linear algebraic calculation yields
\begin{eqnarray}\label{eq:MDB}
 \cM &=& (I - \tau^* J(\FFVS))\inv - (I - \tau^* J(\FFVS))\inv J(\FFVS) \partial_V F_{\phi}(V^*) \nonumber \\ 
 &=& I + (I - \tau^* J(\FFVS))\inv J(\FFVS) (\tau^* I - \partial_V F_{\phi}(V^*)).
\end{eqnarray}
The symmetry of  $J(\FFVS)$ implies that $(I - \tau^* J(\FFVS))\inv J(\FFVS)$ is
also symmetric, which together with the fact that $\tau^* I-\partial_V F_{\phi}(V^*)$ is positive
definite and Lemma \ref{cl:op-prop2} shows that all the eigenvalues
of $ \cM$ are real.  Similar to the proof of Theorem \ref{thm:conv-SCF}, the inequality \eqref{eq:App-Newton-Conv-AN}
 holds if
\begin{eqnarray}\label{eq:pos-lmin}
\lmin( \cM) &>& 0;\\
\label{eq:pos-lmax}
\alpha\lmax( \cM) &<& 2.
\end{eqnarray}

Using $0 > \tau^* >-\frac{1}{\delta}$ and the definition of $\lmin^*$, we have
\begin{eqnarray}\label{eq:tmp1}
\lmin(I - \tau^* J(\FFVS)) &\geq&   \frac{\delta + \lmin^*}{\delta} > 0,\\
\label{eq:tmp2}
\lmax(I - \tau^* J(\FFVS)) &\leq&  \frac{||L^\dagger||_2 + \theta + \delta}{\delta}.
\end{eqnarray}
Using the fact that the smallest eigenvalue of a summation of two matrices is
larger than the summation of the smallest eigenvalues of these matrices, we
obtain
\begin{eqnarray}\label{eq:lminM1-AN}
\lmin( \cM)&\geq& \lmin((I - \tau^* J(\FFVS))\inv)
 + \lmin ((I - \tau^* J(\FFVS))\inv J(\FFVS) (-\partial_V F_{\phi}(V^*)))) \nonumber\\
&\geq&\frac{\delta}{||L^\dagger||_2 + \theta + \delta} +  \lmin ((I - \tau^*
J(\FFVS))\inv J(\FFVS) (-\partial_V F_{\phi}(V^*)))).
\end{eqnarray} 
Applying Lemma \ref{cl:op-prop2}, $\lmax(-\partial_V F_{\phi}(V^*)) \le  \frac{1}{\delta}$ from Lemma
\ref{lm:J-bound} and the definition of $\lmin^*$, we have
\begin{eqnarray}\label{eq:lminM2-AN}
& &\lmin ((I - \tau^* J(\FFVS))\inv J(\FFVS) (-\partial_V F_{\phi}(V^*))))  \nonumber \\
&\geq& \begin{cases} \lmin((I - \tau^* J(\FFVS))\inv)
  \lmin(J(\FFVS))\lmin(-\partial_V F_{\phi}(V^*)), & \mbox{ if } \lmin(J(\FFVS))\geq 0, \\
\lmax((I - \tau^* J(\FFVS))\inv)\lmin(J(\FFVS)) \lmax(-\partial_V F_{\phi}(V^*)), & \mbox{
otherwise}
  \end{cases} \nonumber \\
  &\ge&  \begin{cases} 0, & \mbox{ if } \lmin(J(\FFVS))\geq 0, \\
\frac{\lmin(J(\FFVS))}{\delta + \lmin^*}, & \mbox{ otherwise}
  \end{cases} \nonumber \\ 
  &\ge& \frac{\lmin^*}{\delta + \lmin^*},
\end{eqnarray} 
which together with \eqref{eq:lminM1-AN} gives  \eqref{eq:pos-lmin}.

It follows from Lemma \ref{cl:op-prop2} and \eqref{eq:tmp1} that 
\begin{eqnarray}\label{eq:lmaxM1-AN}
\lmax( \cM) &\leq& \lmax((I - \tau^* J(\FFVS))\inv) + 
 \lmax((I - \tau^* J(\FFVS))\inv \lmax( J(\FFVS)) \lmax(-\partial_V F_{\phi}(V^*)))\nonumber\\
 &\leq& 
   \frac{||L^\dagger||_2 + \theta+\delta}{\delta + \lmin^*}.
\end{eqnarray}
Combining \eqref{eq:alpha-AN2} and \eqref{eq:lmaxM1-AN} together yields \eqref{eq:pos-lmax}.
\end{proof}

Similar to Corollary \eqref{cor-B-psd}, the condition \eqref{eq:conv-cond-AN}
holds when  $J(\FFVS)$ is positive semidefinite.

\subsection{Approximate Newton Method II}

The matrix $J(\rho)$ has to be calculated for each $\rho$ in the approximate
Newton method \eqref{eq:V-App-Newton-iv}. If the computational cost of
second-order derivatives of the
exchange correlation function is expensive, a simpler   
 choice is to  approximate $J(\FFVS)$ by $L^\dagger$ and 
 $\partial_V F_{\phi}(V)$ by $\tau^i I$, that is, $D^i=\tau^i L^\dagger$. Hence,
 approximate Newton method \eqref{eq:V-App-Newton} becomes  
\be\label{eq:V-App-Newton-ii}
V^{i+1} = V^i - \alpha \left(I - \tau^i L^\dagger\right)^{-1}
 \left(V^i - \cV\left(F_\phi(V^i)\right)\right),
\ee
where $\{\tau^i\}$ is negative.
In fact, \eqref{eq:V-App-Newton-ii} is exactly the method of elliptic
preconditioner proposed in \cite{LinYang2012}.

\begin{theorem}\label{thm:conv-SCF-ANL}
Let $V^*$ be a solution of the KS equation \eqref{eq:KS}. Suppose that
Assumption \ref{asmp:exc} holds with a constant $\theta$ and 
Assumption \ref{asmp:UWP} is
valid at $H(V^*)$ with a constant $\delta$ satisfying 
\begin{eqnarray}\label{eq:conv-cond-ANL}
\delta > \theta.
\end{eqnarray}
 Let $\{V^i\}$ be a
sequence generated
by the scheme \eqref{eq:V-App-Newton-iv} using $\lim\limits_{i\rightarrow \infty}\tau_i
=\tau^*\in\left(-\frac{1}{\xi},0\right)$ such that   
$\xi \geq \frac{||L^\dagger||_2 \theta}{\delta - \theta}$,
 and a step size  
\begin{eqnarray}\label{eq:alpha-ANL}
\alpha \in \left(0,\frac{2}{\frac{||L^\dagger||_2 + \xi}{\xi} 
+ \frac{\theta}{\delta}}\right).
\end{eqnarray}
If the initial point $V^0$ is selected in a sufficiently small open neighborhood
 of $V^*$, then $\{V^i\}$  
converges to $V^*$
with R-linear convergence rate no more than 
\begin{eqnarray}\label{eq:rate-ANL}
\max\left\{
\left(1 - \alpha \left(\frac{\xi} {||L^\dagger||_2 + \xi}
- \frac{\theta}{\delta}\right) \right),
\left( \alpha \left(\frac{||L^\dagger||_2 + \xi}{\xi} 
+ \frac{\theta}{\delta}\right) -1 \right)
\right\}.
\end{eqnarray}
\end{theorem}
\begin{proof}
Let $\bar \cM =(I- \tau^* L^\dagger)\inv (I - \partial_V
\cV(F_{\phi}(V^*)))  $.
The convergence of the iteration \eqref{eq:V-App-Newton-ii}
 is guaranteed by 
\begin{eqnarray}\label{eq:App-Newton-Conv-ANL0}
\varrho(I - \alpha \bar \cM) <1.
\end{eqnarray}
 Using the formulation  of $\partial_V
\cV(F_{\phi}(V^*))$, we can decompose $\bar \cM = \bar \cM_1 - \bar \cM_2 $, 
where $\bar \cM_1= (I - \tau^* L^\dagger)\inv (I - L^\dagger \partial_V F_{\phi}(V^*))$ 
and $\bar \cM_2= (I - \tau^* L^\dagger)\inv (J(\FFVS) - L^\dagger ) \partial_V F_{\phi}(V^*)$.
 Since $L^\dagger$ is positive semidefinite, 
a similar proof as Theorem \ref{thm:conv-SCF-AN} implies that 
all the eigenvalues of $\bar \cM_1$
are real and 
\begin{eqnarray}\label{eq:ANL-lmin}
\lmin(\cM_1) &>& \frac{\xi}{||L^\dagger||_2 + \xi},  \\
\label{eq:ANL-lmax}
\lmax(\cM_1) &\leq& \frac{||L^\dagger||_2 + \delta}{\delta}. 
\end{eqnarray}
Using Assumption \ref{asmp:exc} and Lemma \ref{lm:J-bound}, we have
\begin{eqnarray}\label{eq:ANL-M2}
||\bar \cM_2||_2 
&= &  ||(I - \tau^* L^\dagger)\inv (J(\FFVS) - L^\dagger ) \partial_V F_{\phi}(V^*) ||_2 \nonumber\\
&\leq& ||(I - \tau^* L^\dagger)\inv||_2 ||J(\FFVS) - L^\dagger||_2 || \partial_V F_{\phi}(V^*)||_2  \leq \frac{\theta}{\delta}.
\end{eqnarray}
Using \eqref{eq:ANL-lmin} and $\xi \geq \frac{||L^\dagger||_2 \theta}{\delta - \theta}$,
we obtain 
\begin{eqnarray}
\lmin(\bar \cM_1) > \frac{\theta}{\delta},
\end{eqnarray}
which together with \eqref{eq:ANL-M2} yields
\begin{eqnarray}\label{eq:ANL-1}
(1 - \alpha \lmin(\bar \cM_1))  < 1 - \alpha ||\bar \cM_2||_2.
\end{eqnarray}
On the other hand, it follows from \eqref{eq:alpha-ANL}, \eqref{eq:ANL-lmax}
and \eqref{eq:ANL-M2} that 
\begin{eqnarray}\label{eq:ANL-2}
(\alpha \lmax(\bar \cM_1) -1)  < 1 - \alpha ||\bar \cM_2||_2.
\end{eqnarray}
Combining \eqref{eq:ANL-1} and \eqref{eq:ANL-2} together gives
\begin{eqnarray}\label{eq:App-Newton-Conv-ANL}
\varrho (1 - \alpha \bar \cM_1 ) < 1 - \alpha||\bar \cM_2||_2.
\end{eqnarray}
 which guarantees \eqref{eq:App-Newton-Conv-ANL0}.
\end{proof}

\section{Conclusion}
The equivalence between the KS total energy
 minimization problem and the KS equation is ambiguous 
 in the current literatures on KSDFT. A simple counter example shows that 
 the solutions of these two problems are not necessarily the same. We 
 examine the equivalence based on the optimality conditions for a specialized
 exchange correlation functional. We prove that a global
 solution of the KS minimization problem is a solution of 
 the KS equation if the gap between the $p$th and $(p+1)$st eigenvalues of the
 Hamiltonian $H(X)$ is sufficiently large. The equivalence of a local
 minimizer requires that the corresponding charge densities are all positive.
 For strong local minimizers, the nonzero charge densities are bounded below
 by a positive constant uniformly.  These properties are summarized  in Table
 \ref{Table:1}.
 
 \begin{table}[ht]  
   \caption{Equivalence between the KS total energy
    minimization and the KS equation using the exchange correlation function $e\zz \epsilon_{xc}(\rho) = -\frac{3}{4}\gamma \rho\zz\rho^{\frac{1}{3}}$ }
 \begin{center}\tabcolsep=3pt
 \begin{tabular}{|c|c|c|}
 \hline
 \bf properties  & \bf eigenvalue gap $\delta$  & \bf  Other Assumptions\\
 \hline 
  \minitab[c]{A global minimizer $X^*$ \\
     solves\\ 
    the KS equation}
  & \minitab[c]{Assumption \ref{asmp:UWP} holds at $H(X^*)$ with \\

     $\delta > p \ltb||\Ld||_2 - \frac{\gamma}{3}\rtb$}
  & -- \\ 
  \hline 
  \minitab[c]{A local minimizer $X^*$ \\
      solves\\ 
     the KS equation}
  & \minitab[c]{Assumption \ref{asmp:UWP} holds at $H(X^*)$ with\\
    
    $\delta > 2\left(||\Ld||_2 - \frac{\gamma}{3}\right)$}
  & \minitab[c]{$\rho_i > 0$, $i=1,\ldots,n$ }\\ 
  \hline
  \minitab[c]{ $\rho_i(X^*)\in[0,c) \,\Rightarrow$\\ 
        $\rho_i(X^*) = 0$}
  & -- 
  & $X^*$ is a strong local minimizer \\
  \hline
  \end{tabular}
  \end{center}\label{Table:1}
  \end{table}

 We improve the convergence analysis on the SCF iteration
 for solving the KS equation by analyzing the
Jacobian of the corresponding fixed point maps.  
 Global convergence of the simple mixing scheme can be established when there
 exists a gap between $p$th and $(p+1)$st eigenvalues of the
 Hamiltonian $H(X)$.  This assumption can be relaxed for local convergence
 analysis and if the charge density is computed using
 the Fermi-Dirac distribution. Our results
 requires much weaker conditions than the previous analysis in
 \cite{LiuWangWenYuan}. The structure of the Jacobian also suggests two approximate Newton
 methods. In particular, the second one is exactly the method of elliptic
preconditioner proposed in \cite{LinYang2012}.  Although our assumption on the gap is very stringent and is almost
 never satisfied in reality,  our analysis is helpful for a better
 understanding of the KS minimization problem, the KS equation and the SCF
 iteration. A summary of our convergence results is presented in Table \ref{Table:2}.
\begin{table}[ht] 
  \caption{convergence results for solving the KS equation under Assumption \ref{asmp:exc}}
\begin{center}\tabcolsep=2pt
\begin{tabular}{|c|c|c|c|}
\hline
\multicolumn{2}{|c|}{\bf properties}  & \bf \minitab[c]{eigenvalue gap $\delta$
or \\ smoothing parameter $\beta$} 
 & \bf step size $\alpha$ \\
\hline 
\multirow{2}*{
\minitab[c]{\minitab[c]{ analysis of \\ SCF  in \cite{LiuWangWenYuan}}}
} 
 & global convergence 
 & \minitab[c]{UWP holds and\\ 
   $\delta > 12 p\sqrt{n}( \|L^\dagger\|_2 + \theta)$}
 & $1$\\ \cline{2-4} 
 & local convergence 
 & \minitab[c]{Assumption \ref{asmp:UWP} holds at\\
   the local minimizer with\\ 
   $\delta > 2 \sqrt{n} (\|L^\dagger\|_2+ \theta)$}
 & $1$\\ 
\hline
 \multirow{2}*{\minitab[c]{ \\ analysis of \\ SCF with \\
 simple mixing}} 
 & \minitab[c]{global convergence\\ using $F_{\phi}$}
 & \minitab[c]{UWP holds with\\  
   $\delta > ||\Ld||_2 + \theta$}
 & $\left(0, \frac{2\delta}{||\Ld||_2 + \theta + \delta}\right)$\\ \cline{2-4} 
 & \minitab[c]{global convergence\\ using $F_{f_\mu}$}
 &  $\frac{4}{\beta} >  ||\Ld||_2 + \theta$
  & $\left(0,\frac{8}{(||L^\dagger||_2 + \theta)\beta + 4}\right)$\\ \cline{2-4} 
 & \minitab[c]{local  convergence \\using $F_{\phi}$}
 & \minitab[c]{Assumption \ref{asmp:UWP} holds at \\
    the local minimizer with\\  
    $\delta >  -\min\{0, \lambda_{\min}(J(F_\phi(V^*)))\}$}
 & $\left(0, \frac{2\delta}{||L^\dagger||_2 + \theta + \delta}\right)$ \\ \cline{2-4}
  & \minitab[c]{local convergence \\ using $F_{f_\mu}$}
  &   $\frac{4}{\beta} >   -\min\{0, \lambda_{\min}(J(F_\phi(V^*)))\} $ 
   & $\left(0,\frac{8}{(||L^\dagger||_2 + \theta)\beta + 4}\right)$\\ \cline{2-4} 
\hline
 \multirow{3}*{\minitab[c]{\\ analysis of\\ Approximate \\
 Newton \\methods}} 
 & global convergence
 & \minitab[c]{UWP holds with \\ 
    $\delta > \frac{\gamma_{\max}}{\gamma_{\min}}\cdot (||\Ld||_2 + \theta)$}
 &  $\left(0,\frac{2\delta}{||L^\dagger||_2 + \theta + \delta}\right)$ \\ \cline{2-4} 
 & \minitab[c]{local convergence \\on $D^i:=\tau^i J(\rho)$} 
 &  \minitab[c]{Assumption \ref{asmp:UWP} holds at\\
    the local minimizer with\\  
     $\delta >  -\min\{0, \lambda_{\min}(J(F_\phi(V^*)))\}$}
 & $\left(0, \frac{\delta+\lb_{\min}^*}{||L^\dagger||_2 + \theta + \delta}\right)$ \\ \cline{2-4} 
 & \minitab[c]{local convergence \\ on $D^i:=\tau^i \Ld$}
 & \minitab[c]{Assumption \ref{asmp:UWP} holds at\\
     the local minimizer with\\  
     $\delta > \theta$}
 & $\left(0,\frac{2\delta}{\frac{\delta}{\xi}\cdot(||L^\dagger||_2+\xi) +
 \theta}\right)$\\ 
\hline
 \end{tabular}
 \end{center}\label{Table:2} 
 \end{table}

\section*{Acknowledgements}  The authors would like to thank Dr. Chao Yang and
Prof. Aihui Zhou for discussion on the
KS equation and the SCF iteration.

\end{document}